\documentclass[11pt]{article}
\usepackage[utf8]{inputenc}
\usepackage{hyperref}
\usepackage{amsmath}
\usepackage{amsthm}
\usepackage{amssymb}
\usepackage{authblk}
\usepackage{cite}
\usepackage{lipsum}
\usepackage{geometry}
\usepackage{enumerate}
\usepackage{array}
\usepackage{multicol}
\usepackage{multirow}
\usepackage{xcolor}
\newtheorem{theorem}{Theorem}[section]
\numberwithin{theorem}{section}
\newtheorem{lemma}[theorem]{Lemma}
\newtheorem{defi}[theorem]{Definition}

\newtheorem{remark}[theorem]{Remark}

\newcommand{\Max}{\displaystyle \max}
\newcommand{\Sum}{\displaystyle \sum}
\newcommand{\Frac}{\displaystyle \frac}
\newcommand{\F}{\mathbb F}
\newcommand{\Fbn}{\mathbb{F}_{2^n}}
\newcommand{\Fbm}{\mathbb{F}_{2^m}}
\newcommand{\Fp}{\mathbb{F}_{p}}
\newcommand{\Fpn}{\mathbb{F}_{p^n}}

\newcommand{\Fpnmul}{\mathbb{F}_{p^n}^*}

\begin{document}
	\title{On Differential and Boomerang Properties of a Class of Binomials over Finite Fields of Odd Characteristic}
	\author{ Namhun Koo$^1$, Soonhak Kwon$^{2,3}$\\
	\small{\texttt{ Email: komaton@skku.edu, shkwon@skku.edu}}\\
	\small{$^1$Institute of Basic Science, Sungkyunkwan University, Suwon, Korea}\\
	\small{$^2$Department of Mathematics, Sungkyunkwan University, Suwon, Korea}\\
	\small{$^3$Applied Algebra and Optimization Research Center, Sungkyunkwan University, Suwon, Korea}
	}
	\date{}
	
	\maketitle

\begin{abstract}
	In this paper, we investigate the differential and boomerang properties of a class of binomial $F_{r,u}(x) = x^r(1 + u\chi(x))$ over the finite field $\Fpn$, where $r = \frac{p^n+1}{4}$, $p^n \equiv 3 \pmod{4}$, and $\chi(x) = x^{\frac{p^n -1}{2}}$ is the quadratic character in $\Fpn$. We show that $F_{r,\pm1}$ is locally-PN with boomerang uniformity $0$ when $p^n \equiv 3 \pmod{8}$. To the best of our knowledge, it is the second known non-PN function class with boomerang uniformity $0$, and the first such example over odd characteristic fields with $p > 3$. Moreover, we show that $F_{r,\pm1}$ is locally-APN with boomerang uniformity at most $2$ when $p^n \equiv 7 \pmod{8}$. We also provide complete classifications of the differential and boomerang spectra of $F_{r,\pm1}$. Furthermore, we thoroughly investigate the differential uniformity of $F_{r,u}$ for $u\in \Fpnmul \setminus \{\pm1\}$.
	
	
\end{abstract}

\section{Introduction}

Let $p^n$ be an odd prime power and $\Fpn$ be the finite field of $p^n$ elements and $\Fpnmul=\Fpn\setminus \{0\}$ be the multiplicative group of $\Fpn$. Many researchers have been interested in constructing vectorial Boolean functions over finite fields that possess good cryptographic properties. Among these properties, differential uniformity, introduced by Nyberg \cite{Nyb94} is one of the most well-studied and widely recognized criteria due to its strong relevance to resistance against differential cryptanalysis. The differential uniformity is defined as follows.

\begin{defi}
	Let $F$ be a function on $\Fpn$. Let $\delta_F(a,b)$ denote the number of solutions of $F(x+a)-F(x)=b$, where $a\in \Fpnmul$ and $b\in \Fpn$. Then the \textbf{differential uniformity} of $F$ is defined by :
	\begin{equation*}
		\delta_F = \Max_{a\in \Fpnmul, b\in \Fpn}\delta_F(a,b).
	\end{equation*}
	If $\delta_F\le \delta$, then we say that $F$ is differentially $\delta$-uniform. 
\end{defi}
If $F$ is differentially $1$-uniform, then we say $F$ is \emph{perfect nonlinear (PN)}. If $F$ is differentially $2$-uniform, then we say $F$ is \emph{almost perfect nonlinear (APN)}. For a survey of known functions with low differential uniformity, we refer the reader to \cite{Car21}. In recent years, there has been significant progress in studying the differential spectra of various functions. For functions whose differential spectra are known, we refer to tables in the recent results on this topic \cite{MW25,RXY24,XBC+24} and the references therein.

The boomerang attack is a variant of differential cryptanalysis proposed by Wagner \cite{Wag99}. To analyze this type of attack, Cid et al. \cite{CHP+18} introduced the boomerang connectivity table (BCT). The boomerang uniformity is defined as the maximum value among the nontrivial entries in the BCT. This notion was originally defined in \cite{CHP+18} for permutations over binary finite fields. Later, Li et al. \cite{LQSL19} extended the definition to functions that are not necessarily permutations, as follows.

\begin{defi}
	Let $F$ be a function on $\Fpn$. We denote $\beta_F(a,b)$ by the number of common solutions $(x,y)\in \Fpn \times \Fpn$ of the following system :
	\begin{equation*}
		\begin{cases}
			F(x)-F(y)=b,\\
			F(x+a)-F(y+a)=b.
		\end{cases}
	\end{equation*}
	Then the \textbf{boomerang uniformity} of $F$ is defined by :
	\begin{equation*}
		\beta_F = \Max_{a,b\in \Fpnmul}\beta_F(a,b).
	\end{equation*} 
	
\end{defi}
For a survey of known functions with low boomerang uniformity, we refer the reader to \cite{MMM22}. For functions with known boomerang spectra, see Table 1 of \cite{LWZ24} and references therein. 

Most known results on differential or boomerang uniformities, particularly those concerning differential or boomerang spectra, have focused on power functions. Recently, however, several studies have investigated the differential or boomerang properties of functions of the form
\begin{equation*}
	F_{r,u}(x)=x^r (1+u\chi(x)),
\end{equation*}
where $\chi(x) = x^{\frac{p^n - 1}2}$ is the quadratic character in $\Fpn$.
The first result on functions of the above form was introduced by Ness and Helleseth \cite{NH07} that $F_{3^n-2, u}$ is an APN function over $\F_{3^n}$ if $\chi(u-1)=\chi(u+1)=\chi(u)$. Later, Zeng et al. \cite{ZHYJ07} generalized this result by showing that $F_{p^n-2, u}$ is an APN function over $\Fpn$ if $\chi(u-1)=\chi(u+1)=-\chi(5u\pm 3)$, where $p^n\equiv 3\pmod{4}$.

Very recently, two independent studies have been published on the differential properties of $F_{p^n-2, u}$ in cases where it is not APN. In \cite{XBC+24}, Xia et al. studied the differential uniformity of $F_{3^n-2,u}$ over $\F_{3^n}$, and investigated the differential spectrum of $F_{3^n-2,u}$ when $\chi(u+1)=\chi(u-1)$. In \cite{XLB+24}, they generalized their earlier results from \cite{XBC+24} to the case where $p^n\equiv 3\pmod{4}$. On the other hand, Lyu et al. \cite{LWZ24} proved that $F_{p^n-2,u}$ is a differentially $4$-uniform permutation over $\Fpn$ if $\chi(1+u)=\chi(1-u)$, where $p^n \equiv 3 \pmod{4}$. Furthermore, they also showed that $F_{p^n-2,\pm1}$ is a locally-APN function with boomerang uniformity at most $1$, representing the first non-PN class whose boomerang uniformity can attain $0$ or $1$. In particular, $F_{3^n-2, \pm1}$ is locally-PN with boomerang uniformity $0$ over $\F_{3^n}$. They also investigated the differential and boomerang spectra of $F_{p^n-2, \pm1}$. Ren et al. \cite{RXY24} studied the differential spectrum of $F_{p^n-2,u}$ when $\chi(u-1)\ne \chi(u+1)$.

Budaghyan and Pal \cite{BP25} presented experimental results showing that several functions $F_{2,u}$ are APN. They also proved that the differential uniformity of $F_{2,u}$ is at most 5, and conjectured the existence of an infinite APN subclass within this family. Unfortunately, this conjecture was disproved by two studies \cite{MW25,BS25}. Mesnager and Wu \cite{MW25} showed that if $p^n$ is sufficiently large, then the differential uniformity of $F_{2,u}$ is given as follows.
\begin{equation*}
	\delta_{F_{2,u}} =
	\begin{cases}
		\frac{p^n+1}{4} &\text{ if }u\in \{\pm1\},\\
		5 &\text{ if }u\in \Fpn\setminus \mathcal{U}\text{ and }\chi(u+1)=\chi(u-1),\\
		4 &\text{ if }u\in \Fpn\setminus \mathcal{U}\text{ and }\chi(u+1)=-\chi(u-1),\\
		 &\text{ or }p^n\equiv 3\pmod{8}, p
		\ne 3\text{ and }u\in \{\pm \frac{1}{3}\},\\
		3 &\text{ if }p^n\equiv 7\pmod{8} \text{ and }u\in \{\pm \frac{1}{3}\},
	\end{cases}
\end{equation*}
where
\begin{equation*}
	\mathcal{U} = 
	\begin{cases}
		\{0, \pm 1\} &\text{ if }p=3,\\
		\{0, \pm 1, \pm \frac{1}{3}\} &\text{ if }p\ne3.
	\end{cases}
\end{equation*}
In addition, they proved that $F_{2,\pm1}$ is locally-APN with boomerang uniformity at most 2, and further investigated the differential spectrum of $F_{2,\pm1}$. The differential spectrum of $F_{2,\pm1}$ was also studied independently by Yan and Ren \cite{YR25} using a different approach. More recently, Bartoli and Stănică \cite{BS25} disproved the conjecture of \cite{BP25} via function field theory, and extended the nonexistence result to the case of $F_{3,u}$.

In this paper, we study the differential and boomerang properties of the function $F_{r,u}$ in the case
\begin{equation*}
	r=\frac{p^n+1}{4},
\end{equation*}
where $p^n\equiv 3\pmod{4}$. It is known \cite{HRS99} that the power function $F(x)=x^{\frac{p^n+1}{4}}$ is APN when $p^n\equiv 7\pmod{8}$ and $F(x)=x^{\frac{p^n+1}{4}} \chi(x)$ is APN when $p^n\equiv 3\pmod{8}$. The differential spectra of these APN power functions were investigated in \cite{TY23}. In contrast, it is also known \cite{XCX16} that $F(x) = x^{\frac{p^n+1}{4}}$ with $p^n \equiv 3 \pmod{8}$ and $F(x) = x^{\frac{p^n+1}{4}} \chi(x)$ with $p^n \equiv 7 \pmod{8}$ have differential uniformity at most $4$. The differential spectra of these non-APN power functions were later investigated in \cite{BXC+25}. Since these power functions exhibit low differential uniformity in all known cases, it is reasonable to expect that their linear combination $F_{\frac{p^n+1}{4},u}$ also possesses low differential uniformity. In this work, we show that $F_{\frac{p^n+1}{4},u}$ is a differentially $5$-uniform permutation when $\chi(1+u) = (-1)^{\frac{p^n+1}{4}+1} \chi(1-u)$, and $F_{\frac{p^n+1}{4},u}$ is differentially $4$-uniform when $\chi(1+u) = (-1)^{\frac{p^n+1}{4}} \chi(1-u)$. Also, we show that if $u=\pm \frac{1-2^{\frac{p^n+5}{4}}}{3}$, then $F_{\frac{p^n+1}{4}, u}$ is a differentially $4$-uniform permutation, when $p^n\equiv 3 \pmod{8}$. Furthermore, we prove that $F_{\frac{p^n+1}{4},\pm1}$ is locally-PN with boomerang uniformity $0$ when $p^n\equiv 3\pmod{8}$, and is locally-APN with boomerang uniformity at most $2$ when $p^n \equiv 7 \pmod{8}$. To the best of our knowledge, this is the second known non-PN class with boomerang uniformity $0$, and the first such class over odd characteristic fields with $p > 3$. We investigate the differential spectrum of $F_{\frac{p^n+1}{4},\pm 1}$. We also study the boomerang spectrum of $F_{\frac{p^n+1}{4},\pm 1}$, and show that $F_{\frac{p^n+1}{4},\pm 1}$ has boomerang uniformity $2$ if $p^n\ne 7, 31$,  when $p^n \equiv 7 \pmod{8}$.

The remainder of this paper is organized as follows. Section \ref{sec_pre} contains some preliminaries. In Section \ref{sec u=pm1}, we study the differential and boomerang spectra of $F_{\frac{p^n+1}{4},\pm 1}$. In Section \ref{sec u!=pm1}, we study the differential uniformity of $F_{\frac{p^n+1}{4},u}$, where $u\in \Fpnmul \setminus \{\pm1\}$. Finally, we give a concluding remark in Section \ref{sec_con}.

\section{Preliminaries}\label{sec_pre}

If $F(x)=x^d$ is a power function, then $b=F(x+a)-F(x)=(x+a)^d-x^d$ is equivalent to
\begin{equation*}
	\Frac{b}{a^d}=\left(\Frac x a +1\right)^d -\left(\Frac x a\right)^d = (y+1)^d-y^d,
\end{equation*}
where the last equality is from setting $y=\Frac x a$. Thus, we have 
\begin{equation}\label{dupower_eqn}
	\delta_F(a,b)=\delta_F\left(1, \Frac{b}{a^d} \right), \text{ so }	\delta_F=\Max_{b\in \Fpn}\delta_F(1,b),
\end{equation}
in this case. In \cite{BCC11}, Blondeau et al. introduced a new notion \emph{locally-APN} function when $p=2$. In \cite{HLX+23}, Hu et al. generalized the notion of locally-APN power function for all primes $p$, as follows.
 
\begin{equation}\label{locallyapn_def}
	\delta_F(1,b)\le 2\text{ for all }b\in \Fpn\setminus\Fp.
\end{equation}

In \cite{LWZ24} and \cite{MW25}, the authors discussed the locally-APNness of $F_{r,\pm1}$ when $r=p^n-2$ and $r=2$, respectively. They defined that a function $F$ on $\Fpn$ is locally-APN if
\begin{equation}\label{locallyapn_wrongdef} 
	\delta_F(a,b)\le 2\text{ for all }a\in \Fpnmul\text{ and }b\in \Fpn\setminus \Fp,
\end{equation} 
and claimed that $F_{r,\pm1}$ satisfy this property and hence it is locally-APN. Note that they actually showed $\delta_{F_{r,\pm1}}(a,b)\le 2$ for all $a,b\in\Fpnmul$. In our view, this definition \eqref{locallyapn_wrongdef} is stronger than the original definition \eqref{locallyapn_def}. For example, if $F(x) = x^{2^m-1}$ on $\Fbn$ with $n\equiv2\pmod{4}$ and $n=2m$, then $F$ is locally-APN with $\delta_F(1,1)=4$ (see Theorem 7 of \cite{BCC11}). By \eqref{dupower_eqn} we can see that $4=\delta_F(1,1)=\delta_F(a,a^{2^m-1})$. Since $a^{2^m-1}\ne 1$ when $a\in \Fbn\setminus\Fbm\subset \Fbn\setminus \F_2$, we can see that $F$ does not satisfy \eqref{locallyapn_wrongdef}. This shows that extension of the definition locally-APNness for power functions \eqref{locallyapn_def} to the case of all functions needs to be made with more care. 

However, we can naturally extend \eqref{locallyapn_def} to functions satisfying the following conditions :
\begin{equation}\label{locallyapn_cordef_eqn}
	\delta_F(a,b)= \delta_F(1,g_a(b))\text{ for all }b\in \Fpn,
\end{equation}
where $g_a$ permutes $\Fpn$ for each $a\in\Fpnmul$. By \eqref{dupower_eqn}, power functions also satisfy \eqref{locallyapn_cordef_eqn}.

\begin{defi}\label{locallyapn_cordef}
	Let $F$ be a function on $\Fpn$ satisfying \eqref{locallyapn_cordef_eqn} for every $a\in \Fpnmul$. Then, 
	\begin{itemize}
		\item $F$ is called \textbf{locally-PN} if $\delta_F(1,b)\le 1$ for all $b\in \Fpn\setminus \Fp$.
		\item $F$ is called \textbf{locally-APN} if $\delta_F(1,b)\le 2$ for all $b\in \Fpn\setminus \Fp$.
	\end{itemize}
\end{defi}

Note that $F_{r,u}$ satisfy \eqref{locallyapn_cordef_eqn} (see Lemma \ref{Fruproperty2_lemma}) and the authors of \cite{LWZ24} and \cite{MW25} showed that $\delta_{F_{r,\pm1}}(1,b)\le 2$ for all $b\in \Fpnmul$, when $r=p^n-2$ and $r=2$, respectively. So, we can see that these results in \cite{LWZ24} and \cite{MW25} are also valid under the above new definition for locally-APN functions. Furthermore, according to \cite{LWZ24}, $\delta_{F_{p^n-2,\pm1}}(1,b)\le 1$ for all $b\in \Fpnmul$ when $p=3$, so we can say that  $F_{3^n-2,\pm1}$ is a locally-PN function on $\F_{3^n}$ when $n$ is odd.

\begin{lemma}[Lemma 10 of \cite{MW25}]\label{Fruproperty1_lemma}
	Let $a\in \Fpnmul$ and $b\in \Fpn$. Then, we have $\delta_{F_{r,-u}}(a,-b)=\delta_{F_{r,u}}(a,(-1)^{r+1}b)$ and $\beta_{F_{r,-u}}(a,-b)=\beta_{F_{r,u}}(a,(-1)^{r}b)$, and hence $F_{r,u}$ and $F_{r,-u}$ has the same differential and boomerang spectrum.
\end{lemma}

\begin{lemma}[Lemma 11 of \cite{MW25}]\label{Fruproperty2_lemma}
	Let $a\in \Fpnmul$ and $b\in \Fpn$. Then,
	\begin{align*}
		\delta_{F_{r,u}}(a,b)&=
		\begin{cases}
			\delta_{F_{r,u}}\left( 1, \frac{b}{a^r}\right) &\text{ if }\chi(a)=1,\\
			\delta_{F_{r,u}}\left( 1, \frac{b}{(-1)^{r+1}a^r}\right) &\text{ if }\chi(a)=-1,
		\end{cases}
		\\
		\beta_{F_{r,u}}(a,b)&=
		\begin{cases}
			\beta_{F_{r,u}}\left( 1, \frac{b}{a^r}\right) &\text{ if }\chi(a)=1,\\
			\beta_{F_{r,u}}\left( 1, \frac{b}{(-1)^{r}a^r}\right) &\text{ if }\chi(a)=-1.
		\end{cases}
	\end{align*}
\end{lemma}

It is known \cite{LWZ24} that $F_{p^n-2,u}$ is differentially $4$-uniform permutation when $\chi(1+u)=\chi(u-1)$. In the following, we discuss on the condition that $F_{r,u}$ is a permutation polynomial. 

\begin{lemma}[\cite{PL01}]\label{PP_lemma}
	 Let $r,s$ be integers with $s\mid p^n-1$. Then, $f(x)=x^r h(x^s)$ is a permutation polynomial if and only if the followings hold :
	\begin{itemize}
		\item $\gcd(r,s)=1$,
		\item $g(x)=x^r \left( h(x)\right)^s$ permutes $H=\langle \xi^s\rangle$, where $\xi$ is a primitive element in $\Fpn$. 
	\end{itemize}
\end{lemma}

Now we apply $s=\frac{p^n-1}{2}$ and $h(x)=1+ux$ on Lemma \ref{PP_lemma}. Then, $H=\{-1,1\}$ and $g(1)=\chi(1+u)$ and $g(-1)=(-1)^r\chi(1-u)$. So, we have the following.

\begin{theorem}\label{Fru_PP_thm}
	If $\gcd\left( r,\frac{p^n-1}{2}\right) \ne 1$, then $F_{r,u}$ is not a PP. If $\gcd\left( r,\frac{p^n-1}{2}\right) = 1$, then $F_{r,u}$ is a PP if and only if $\chi(1+u)\ne (-1)^r \chi(1-u)$.	
\end{theorem}

Note that $F_{2,u}$ is a differentially $5$-uniform permutation if $\chi(1+u)=\chi(u-1)$, by the above theorem and \cite{MW25}.

The following well-known lemmas are useful to compute the differential and boomerang spectra of $F_{r,1}$. We denote $S_0= \{x\in \Fpn : \chi(x)=1\}$ and $S_1=\{x\in \Fpn : \chi(x)=-1\}$. Moreover,
\begin{align*}
	S_{00}&=\{x\in \Fpn : \chi(x)=\chi(x+1)=1\},\\
	S_{01}&=\{x\in \Fpn : \chi(x)=1, \chi(x+1)=-1\},\\
	S_{10}&=\{x\in \Fpn : \chi(x)=-1, \chi(x+1)=1\},\\
	S_{11}&=\{x\in \Fpn : \chi(x)=\chi(x+1)=-1\}.
\end{align*}

\begin{lemma}[\cite{Dic35}]\label{Sset num of elts} If $p^n\equiv 1\pmod{4}$, then $\#S_{00}=\frac{p^n-5}{4}$ and $\#S_{01}=\#S_{10}=\#S_{11}=\frac{p^n-1}{4}$. If $p^n\equiv 3\pmod{4}$, then $\#S_{00}=\#S_{10}=\#S_{11}=\frac{p^n-3}{4}$ and $\#S_{01}=\frac{p^n+1}{4}$.
\end{lemma}

\begin{lemma}[Theorem 5.48 of \cite{LN97}]\label{chi_lemma quad}
	Let $f(x) = a_2x^2 +a_1 x + a_0 \in \Fpn [x]$ with $p$ odd and $a_2 \ne 0$. Put $d=a_1^2-4a_0a_2$. Then,
	\begin{equation*}
		\sum_{x\in \Fpn} \chi(f(x)) = 
		\begin{cases}
			-\chi(a_2) &\text{ if }d\ne 0,\\
			(p^n-1)\chi(a_2) &\text{ if }d = 0.
		\end{cases}
	\end{equation*}
\end{lemma}

\begin{lemma}[Theorem 5.41 of \cite{LN97}]\label{chi_lemma inequal}
	Let $\psi(\cdot)$ be a multiplicative character of $\Fpn$ of order $m>1$ and let $f\in\Fpn[x]$ be a monic polynomial of positive degree that is not $m$-th power of a polynomial. Let $d$ be the number of distinct roots of $f$ in its splitting field over $\Fpn$. Then for every $a\in\Fpn$ we have
	\begin{equation*}
		\left| \sum_{x\in \Fpn} \psi (af(x))\right| \le (d-1)\sqrt{p^n}.
	\end{equation*}
\end{lemma}

\section{Differential and Boomerang Spectra when $u = \pm 1$}\label{sec u=pm1}

Throughout the remainder of this paper, we fix 
\begin{equation*}
	p^n\equiv 3\pmod{4}\text{ and }r=\frac{p^n+1}{4},
\end{equation*}
and omit explicit mention of $r$ unless necessary. For example, we write $F_{r,u}$ instead of $F_{\frac{p^n+1}{4},u}$, for simplicity of notation. We also frequently use, without further mention, the identities $\chi(2) = (-1)^r$ and $x^{2r}=x\cdot \chi(x)$. 

In this section, we study the differential and boomerang spectra of $F_{r,u}$ when $u=\pm 1$. By Lemma \ref{Fruproperty1_lemma}, it is enough to consider the case $u=1$.

\subsection{Differential Spectrum of $F_{r,1}$}\label{subsec DS}

In this subsection, we study the differential spectrum of $F_{r,1}$. We desire to count the number of solutions of 
\begin{equation}\label{DU_eqn u=1}
	b=F(x+1)-F(x) = (x+1)^r(1+\chi(x+1))-x^r(1+\chi(x)).
\end{equation}
Denote $D_{ij}^1(b)$ be the number of solutions of \eqref{DU_eqn u=1} in $S_{ij}$ where $i,j\in\{0,1\}$. 

\begin{lemma}\label{DU_lemma u=1 S11}
	We have 
	\begin{equation*}
		D_{11}^1(b)=
		\begin{cases}
			\frac{p^n-3}{4}, &\text{ if }b=0,\\
			0, &\text{ otherwise.}
		\end{cases}
	\end{equation*}
\end{lemma}

\begin{proof}
	If $x\in S_{11}$, then \eqref{DU_eqn u=1} is equivalent to $b=0$.
	By Lemma \ref{Sset num of elts}, we have $\#S_{11}=\frac{p^n-3}{4}$, which completes the proof.
\end{proof}

\begin{lemma}\label{DU_lemma u=1 S00}
	We obtain
	\begin{equation*}
		D_{00}^1(b)=
		\begin{cases}
			1, &\text{ if }\chi\left(b(b^2+4)\right)=1\text{ and }\chi\left( b(b^2-4)\right) =-1,\\
			0, &\text{ otherwise.}
		\end{cases}
	\end{equation*}
\end{lemma}

\begin{proof}
	If $x\in S_{00}$, then \eqref{DU_eqn u=1} becomes
	\begin{equation}\label{DU_eqn1 u=1 S00}
		\frac{b}{2} = (x+1)^r-x^r.
	\end{equation}
	If $b=0$, then squaring on both sides of \eqref{DU_eqn1 u=1 S00} implies $x+1=x$, a contradiction. Hence, we have $D_{00}^1(0)=0$. From now, we assume that $b\ne 0$. Then, \eqref{DU_eqn1 u=1 S00} implies that
	\begin{equation}\label{DU_eqn2 u=1 S00}
		\frac{2}{b} = \frac{1}{(x+1)^r - x^r} = (x+1)^r + x^r.
	\end{equation}
	From \eqref{DU_eqn1 u=1 S00} and \eqref{DU_eqn2 u=1 S00}, we obtain
	\begin{equation}\label{DU_eqn3 u=1 S00}
		x^r = \frac{1}{2} \left( \frac{2}{b} - \frac{b}{2} \right) = \frac{4 - b^2}{4b}, \ \ (x+1)^r = \frac{1}{2} \left( \frac{2}{b} + \frac{b}{2} \right) = \frac{4 + b^2}{4b}.
	\end{equation}
	Squaring on the both sides of \eqref{DU_eqn3 u=1 S00}, we have
	\begin{equation}\label{DU_eqn u=1 S00 x}
		x=\frac{(4-b^2)^2}{16b^2}, \ \ x+1=\frac{(4-b^2)^2}{16b^2}+1=\frac{(4+b^2)^2}{16b^2},
	\end{equation}
	and hence we can see that $\chi(x)=\chi(x+1)=1$. Substituting \eqref{DU_eqn u=1 S00 x} in \eqref{DU_eqn1 u=1 S00} we obtain
	\begin{align}
		\frac{b}{2} &= \left( \frac{(b^2+4)^2}{16b^2}\right)^r -  \left( \frac{(b^2-4)^2}{16b^2}\right)^r =  \chi\left( \frac{b^2+4}{4b}\right)\frac{b^2+4}{4b} - \chi\left( \frac{b^2-4}{4b}\right)\frac{b^2-4}{4b} \notag\\
		&=\frac{b}{4}\left( \chi\left( b(b^2+4)\right)-\chi\left( b(b^2-4)\right) \right) + \frac{1}{b}\left( \chi\left( b(b^2+4)\right)+\chi\left( b(b^2-4)\right) \right).  \label{DU_eqn4 u=1 S00}
	\end{align}
	One can easily verify that \eqref{DU_eqn4 u=1 S00} holds if and only if $\chi\left(b(b^2+4)\right)=1$ and $\chi\left(b(b^2-4)\right)=-1$, while all other cases lead to a contradiction.
\end{proof}

\begin{lemma}\label{DU_lemma u=1 chi(x)!=chi(x+1)}
	Let $b\ne 0$. Then, 
	\begin{align*}
		D_{01}^1(b)&=
		\begin{cases}
			1 &\text{ if }\chi(b)=-\chi(2)\text{ and }\chi\left( b^2+4\right) =-1,\\
			0 &\text{ otherwise,}
		\end{cases}
		\\ 
		D_{10}^1(b)&=
		\begin{cases}
			1 &\text{ if }\chi(b)=\chi(2)\text{ and }\chi\left( b^2-4 \right) =-1,\\
			0 &\text{ otherwise.}
		\end{cases}
	\end{align*}
	Furthermore, \eqref{DU_eqn u=1} has at most one solution in $S_{01}\cup S_{10}$.
\end{lemma}
\begin{proof}
	If $x\in S_{01}$, then \eqref{DU_eqn u=1} becomes $b=-2x^r$, which is equivalent to  
	\begin{equation}\label{DU_eqn1 u=1 S01}
		x^r=-\frac{b}{2}.
	\end{equation}
	Squaring on the both sides of \eqref{DU_eqn1 u=1 S01}, we have
	\begin{equation}\label{DU_eqn2 u=1 S01}
		x=\frac{b^2}{4}.
	\end{equation}
	Substituting \eqref{DU_eqn2 u=1 S01} to \eqref{DU_eqn1 u=1 S01}, we have
	\begin{equation*}
		-\frac{b}{2}=\left( \frac{b^2}{4} \right)^r = \chi\left( \frac{b}{2}\right) \cdot \frac{b}{2},
	\end{equation*}
	which is equivalent to $\chi(b)=-\chi(2)$. Since $x+1=\frac{b^2+4}{4}$, we have \eqref{DU_eqn u=1} has one solution in $S_{01}$ if and only if $\chi(b)=-\chi(2)$ and $\chi(b^2+4)=-1$.
	
	The proof for $D_{10}^1(b)$ is similar to $D_{01}^1(b)$, and we omit it here. Since $\chi(b)=\chi(2)$ and $\chi(b)=-\chi(2)$ cannot hold simultaneously, \eqref{DU_eqn u=1} has at most one solution in $S_{01}\cup S_{10}$.
\end{proof}

For any function $F$ satisfying \eqref{locallyapn_cordef_eqn}, the differential spectrum of $F$ is defined to be the multiset $DS_F = \{\omega_i : 0\le i \le \delta_F\}$, where
\begin{equation*}
	\omega_i = \#\{b\in \Fpn : \delta_F(1,b)=i\}. 
\end{equation*} 
The following identity for the differential spectrum is well-known :
\begin{equation}\label{DS_identity}
	\sum_{i=0}^{\delta_F}\omega_i = \sum_{i=0}^{\delta_F}i\cdot \omega_i = p^n.
\end{equation}

The following lemma is used in our main theorem of this subsection to compute the differential spectrum of $F_{r,1}$.

\begin{lemma}[Lemma 7 of \cite{MW25}]\label{chi_lemma x^4-1}
	If $p^n\equiv 3\pmod{4}$, then 
	\begin{equation*}
		\sum_{x\in \Fpn} \chi(x^4-1)=-1.
	\end{equation*}
\end{lemma}

Now we are ready to prove the main theorem of this subsection.
\begin{theorem}\label{DS_theorem}
	If $p^n\equiv 3\pmod{8}$, then $F_{r,1}$ is locally-PN, and the differential spectrum of $F_{r,1}$ is given by
	\begin{equation*}
		DS_{F_{r,1}} = \left\{\omega_0 = \frac{p^n-3}{4},\ \omega_1 = \frac{3p^n-1}{4},\ \omega_{\frac{p^n+1}{4}}=1\right\}.
	\end{equation*}
	If $p^n\equiv 7\pmod{8}$, then $F_{r,1}$ is locally-APN, and the differential spectrum of $F_{r,1}$ is given by
	\begin{equation*}
		DS_{F_{r,1}} = \left\{\omega_0 = \frac{p^n-3}{2},\ \omega_1 = \frac{p^n+5}{4},\ \omega_2 = \frac{p^n-3}{4},\  \omega_{\frac{p^n+1}{4}}=1\right\}.
	\end{equation*}
\end{theorem}

\begin{proof}
	When $x=0$, we have $b=2$ from \eqref{DU_eqn u=1}. When $x=-1$, we have $b=0$ from \eqref{DU_eqn u=1}.\\
	By Lemmas \ref{DU_lemma u=1 S11}, \ref{DU_lemma u=1 S00} and \ref{DU_lemma u=1 chi(x)!=chi(x+1)}, we have $D_{00}^1(2)=D_{01}^1(2)=D_{10}^1(2)=D_{11}^1(2)=0$, so $$\delta_{F_{r,1}}(1,2)=1.$$
	Similarly, we have $D_{00}^1(0)=D_{01}^1(0)=D_{10}^1(0)=0$ and $D_{11}^1(0)=\frac{p^n-3}{4}$. So, 
	$$\delta_{F_{r,1}}(1,0) = \#S_{11} +1 =\frac{p^n-3}{4}+1=\frac{p^n+1}{4}.$$
	If $b\ne 0,2$, then \eqref{DU_eqn u=1} has a solution in $S_{00}\cup S_{01} \cup S_{10}$. By Lemma \ref{DU_lemma u=1 S00}, \eqref{DU_eqn u=1} has at most one solution in $S_{00}$.  By Lemma \ref{DU_lemma u=1 chi(x)!=chi(x+1)}, \eqref{DU_eqn u=1} has at most one solution in $S_{01} \cup S_{10}$. Thus, \eqref{DU_eqn u=1} has at most two solutions in $S_{00}\cup S_{01} \cup S_{10}$, and hence we have $\delta_{F_{r,1}}(1,b)\le 2$ when $b\not\in \{0,2\}$. Therefore, we have $\delta_{F_{r,1}}(1,b)\le 2$ for all $b\in \Fpnmul$ and hence $F_{r,1}$ is locally-APN.
	
	If $p^n\equiv 3\pmod{8}$, then $\chi(2)=-1$. We show that if \eqref{DU_eqn u=1} has a solution in $S_{01} \cup S_{10}$ then \eqref{DU_eqn u=1} has no solution in $S_{00}$. If $D_{01}^1(b)=1$, then $\chi(b)=1$ and $\chi(b^2+4)=-1$ and $D_{10}^1(b)=0$, by Lemma \ref{DU_lemma u=1 chi(x)!=chi(x+1)}. Then we have $\chi\left( b(b^2+4) \right) =-1$, and hence $D_{00}^1(b)=0$ by Lemma \ref{DU_lemma u=1 S00}. Similarly, we can show that if $D_{10}^1(b)=1$ then $D_{00}^1(b)=D_{01}^1(b)=0$. Therefore, we have $\delta_{F_{r,1}}(1,b)\le 1$ for all $b\in \Fpnmul$ in any cases, and hence $F_{r,1}$ is locally-PN. Applying $\omega_{\frac{p^n+1}{4}}=1$ on \eqref{DS_identity}, we have $\omega_1=\frac{3p^n-1}{4}$ and $\omega_0 = \frac{p^n-3}{4}$.
	
	If $p^n\equiv 7\pmod{8}$, then $\chi(2)=1$. By Lemma \ref{DU_lemma u=1 S00} and Lemma \ref{DU_lemma u=1 chi(x)!=chi(x+1)}, $\delta_F(1,b)=2$ if and only if one of the following two conditions holds
	\begin{equation*}
		\begin{cases}
			\chi(b)=-1,\ \chi(b^2+4)=-1,\ \chi(b^2-4)=1,\\
			\chi(b)=1,\ \chi(b^2+4)=1,\ \chi(b^2-4)=-1.
		\end{cases}
	\end{equation*}
	By Lemma \ref{chi_lemma x^4-1}, we have
	\begin{equation}\label{chi x^4-16}
		\sum_{x\in \Fpn}\chi(x^4-16) = \sum_{x\in \Fpn}\chi(16)\chi\left(\left(\frac{x}{2}\right)^4-1\right) =  \sum_{y\in \Fpn}\chi(y^4-1)=-1.
	\end{equation}
	Moreover,
	\begin{equation*}
		\sum_{x\in \Fpn}\chi\left(x(x^2\pm4)\right) = \sum_{y\in \Fpn}\chi\left(-y\left((-y)^2\pm4\right)\right)  = -\sum_{y\in \Fpn}\chi\left(y\left(y^2\pm4\right)\right) = -\sum_{x\in \Fpn}\chi\left(x\left(x^2\pm4\right)\right)
	\end{equation*}
	implies that
	\begin{equation*}
		\sum_{x\in \Fpn}\chi\left(x(x^2\pm4)\right) = 0.
	\end{equation*}
	Hence,
	\begin{align*}
		\omega_2&= \frac{1}{8}\sum_{x\in \Fpn\setminus\{0,\pm2\}}\left( (1-\chi(x))(1-\chi(x^2+4))(1+\chi(x^2-4)) + (1+\chi(x))(1+\chi(x^2+4))(1-\chi(x^2-4))\right)\\
		&=\frac{1}{8}\sum_{x\in \Fpn}\left(2 + 2\chi\left(x(x^2+4) \right) - 2\chi\left(x(x^2-4) \right) - 2\chi\left((x^2+4)(x^2-4) \right)\right)-1\\
		&=\frac{1}{4}\sum_{x\in \Fpn}\left(1  - \chi\left(x^4-16 \right)\right) -1= \frac{p^n-3}{4}.
	\end{align*}
	Applying $\omega_{\frac{p^n+1}{4}}=1$ and $\omega_2=\frac{p^n-3}{4}$ on \eqref{DS_identity}, we have $\omega_1=\frac{p^n+5}{4}$ and $\omega_0 = \frac{p^n-3}{2}$. We complete the proof.
\end{proof}

We confirm that the above theorem is true for $7< p^n < 100000$ via SageMath. Table \ref{table_ds} describes $DS_{F_{r,1}}$ for $p^n < 200$, which is consistent with the result established in Theorem \ref{DS_theorem}. Note that $F_{r,1}$ has the same differential spectrum with $F_{3^n-2, 1}$, when $p=3$.

\begin{table}
	\begin{center}
	\begin{tabular}{cc|cc}
		\hline $p^n$\ & $DS_{F_{r,1}}$ &  $p^n$\ & $DS_{F_{r,1}}$ \\
		\hline $3$ & $\{ \omega_1 = 3 \}$ & $7$ & $\{\omega_0 = 2,\ \omega_1 = 3,\ \omega_2 = 2 \}$\\
		$27$ & $\{\omega_0 = 6,\ \omega_1 = 20,\ \omega_7 = 1 \}$ & $23$ & $\{\omega_0 = 10,\ \omega_1 = 7,\ \omega_2 = 5,\ \omega_{6}=1 \}$ \\
		$11$ & $\{\omega_0 = 2,\ \omega_1 = 8,\ \omega_3 = 1 \}$ & $31$ & $\{\omega_0 = 14,\ \omega_1 = 9,\ \omega_2 = 7,\ \omega_{8}=1 \}$ \\
		$19$ & $\{\omega_0 = 4,\ \omega_1 = 14,\ \omega_5 = 1 \}$ & $47$ & $\{\omega_0 = 22,\ \omega_1 = 13,\ \omega_2 = 11,\ \omega_{12}=1 \}$ \\
		$43$ & $\{\omega_0 = 10,\ \omega_1 = 32,\ \omega_{11} = 1 \}$ & $71$ & $\{\omega_0 = 34,\ \omega_1 = 19,\ \omega_2 = 17,\ \omega_{18}=1 \}$ \\
		$59$ & $\{\omega_0 = 14,\ \omega_1 = 44,\ \omega_{15} = 1 \}$ & $79$ & $\{\omega_0 = 38,\ \omega_1 = 21,\ \omega_2 = 19,\ \omega_{20}=1 \}$ \\
		$67$ & $\{\omega_0 = 16,\ \omega_1 = 50,\ \omega_{17} = 1 \}$ & $103$ & $\{\omega_0 = 50,\ \omega_1 = 27,\ \omega_2 = 25,\ \omega_{26}=1 \}$ \\
		$83$ & $\{\omega_0 = 20,\ \omega_1 = 62,\ \omega_{21} = 1 \}$ & $127$ & $\{\omega_0 = 62,\ \omega_1 = 33,\ \omega_2 = 31,\ \omega_{32}=1 \}$ \\
		$107$ & $\{\omega_0 = 26,\ \omega_1 = 80,\ \omega_{27} = 1 \}$ & $151$ & $\{\omega_0 = 74,\ \omega_1 = 39,\ \omega_2 = 37,\ \omega_{38}=1 \}$ \\
		$131$ & $\{\omega_0 = 32,\ \omega_1 = 98,\ \omega_{33} = 1 \}$ & $167$ & $\{\omega_0 = 82,\ \omega_1 = 43,\ \omega_2 = 41,\ \omega_{42}=1 \}$ \\
		$139$ & $\{\omega_0 = 34,\ \omega_1 = 104,\ \omega_{35} = 1 \}$ & $191$ & $\{\omega_0 = 94,\ \omega_1 = 49,\ \omega_2 = 47,\ \omega_{48}=1 \}$ \\
		$163$ & $\{\omega_0 = 40,\ \omega_1 = 122,\ \omega_{41} = 1 \}$ & $199$ & $\{\omega_0 = 98,\ \omega_1 = 51,\ \omega_2 = 49,\ \omega_{50}=1 \}$ \\
		$179$ & $\{\omega_0 = 44,\ \omega_1 = 134,\ \omega_{45} = 1 \}$ & &  \\
		\hline
	\end{tabular}
	\caption{Differential spectrum $DS_{F_{r,1}}$ when $p^n < 200$.}\label{table_ds}
	\end{center}
\end{table}

\subsection{Boomerang Spectrum of $F_{r,1}$}\label{subsec BS}

In this subsection, we study the boomerang spectrum of $F_{r,1}$. We consider to find the number of common solutions $(x,y)$ of the following system.

\begin{equation}\label{BU_system}
	\begin{cases}
		x^r(1+\chi(x))-y^r(1+\chi(y))=b,\\
		(x+1)^r(1+\chi(x+1))-(y+1)^r(1+\chi(y+1))=b.
	\end{cases}
\end{equation}

\begin{theorem}\label{BU_thm 3mod8}
	If $p^n\equiv3\pmod{8}$, then $\beta_{F_{r,\pm1}}=0$. 
\end{theorem}
\begin{proof}
	By Lemma \ref{Fruproperty1_lemma}, it is enough to only consider the case $u=1$. Suppose on the contrary that there is a solution $(x,y)=(x_0,y_0)$ of \eqref{BU_system} with $x_0\ne y_0$ and $b\ne 0$. Since
	\begin{equation}\label{BU_eqn 3mod8}
		F_{r,1}(x_0)-F_{r,1}(y_0)=F_{r,1}(x_0+1)-F_{r,1}(y_0+1)=b,
	\end{equation} 
	we have 
	$$F_{r,1}(x_0+1)-F_{r,1}(x_0) = F_{r,1}(y_0+1)-F_{r,1}(y_0).$$
	Hence $x_0$ and $y_0$ are two distinct solutions of $F_{r,1}(x+1)-F_{r,1}(x)=c$ for some $c\in \Fpn$. By Theorem \ref{DS_theorem}, $\delta_{F_{r,1}} (1,c)>1$ implies $c=0$. According to the discussion in Section \ref{subsec DS}, we obtain $x_0, y_0 \in S_{11}\cup\{-1\}$. Then, we have $$F_{r,1}(x_0)=F_{r,1}(x_0+1)=F_{r,1}(y_0)=F_{r,1}(y_0+1)=0.$$
	and hence $b=0$ by \eqref{BU_eqn 3mod8}, a contradiction.
	
	Therefore, there is no solution of \eqref{BU_system}, and hence $\beta_{F_{r,1}}=0$. 
\end{proof}

\bigskip
In the rest of this subsection, we study the boomerang spectrum of $F_{r,1}$, when $p^n \equiv 7 \pmod{8}$. Denote $B_{ijkl}(b)$ be the number of solutions of \eqref{BU_system} in $S_{ij}\times S_{kl}$ where $i,j,k,l\in\{0,1\}$.

\begin{lemma}\label{BU_lemma S11}
	Let $b\in \Fpnmul$. Then, there is no solution of \eqref{BU_system} satisfying $x\in S_{11}\cup\{0,-1\}$ or $y\in S_{11}\cup\{0,-1\}$. 
\end{lemma}
\begin{proof}
	We only show that there is no solution with $x\in S_{11}\cup\{0,-1\}$. To show that there is no solution with $y\in S_{11}\cup\{0,-1\}$ is similar.
	
	Suppose on the contrary that there is a solution $(x,y)$ of \eqref{BU_system} such that $x\in S_{11}\cup\{-1\}$. Then, $F_{r,1}(x)=F_{r,1}(x+1)=0$, and hence \eqref{BU_system} corresponds to 
	\begin{equation}\label{BU_eqn1 S11}
		-y^r(1+\chi(y))=-(y+1)^r(1+\chi(y+1))=b.
	\end{equation}
	If $y\in S_{00}$, then \eqref{BU_eqn1 S11} leads to $y^r =(y+1)^r$, which implies $y=y+1$, a contradiction. If $y\not \in S_{00}$, then we have $F_{r,1}(y)=0$ or $F_{r,1}(y+1)=0$, so $b = F_{r,1}(x) - F_{r,1}(y) = F_{r,1}(x+1) - F_{r,1}(y+1) = 0$, a contradiction. Therefore, there is no solution of \eqref{BU_system} satisfying $x\in S_{11}\cup\{-1\}$.
	
	It remains to consider the special case $x=0$.
	Suppose that $(0,y)$ is a solution of \eqref{BU_system}.
	Then, \eqref{BU_system} implies
	$$
	-y^r(1+\chi(y))=2-(y+1)^r(1+\chi(y+1))=b,
	$$
	which yields $F_{r,1}(y+1)-F_{r,1}(y)=2$.
	However, we have already shown in the proof of Theorem~\ref{DS_theorem}
	that this equation admits the unique solution $y=0$.
	Consequently, we obtain $b=0$, a contradiction.
	Therefore, there is no solution of \eqref{BU_system} with $x=0$.
\end{proof}

\begin{lemma}\label{BU_lemma nosol}
	Let $b\in \Fpnmul$. Then, $B_{0000}(b)=B_{0101}(b)=B_{0110}(b)=B_{1010}(b)=B_{1001}(b)=0$. 
\end{lemma}

\begin{proof}
If $(x,y)\in S_{00}\times S_{00}$, then equations in \eqref{BU_system} correspond to 
\begin{equation*}
	\begin{cases}
		2x^r-2y^r=b,\\
		2(x+1)^r-2(y+1)^r=b.
	\end{cases}
\end{equation*}
which implies that $x^r-y^r=(x+1)^r-(y+1)^r$. Multiplying $(x^r+y^r)\left( (x+1)^r+(y+1)^r\right)$ on the both sides of the last equation, we have
$x=y$ or $x^r+y^r=(x+1)^r+(y+1)^r$. If $x=y$, then we have $b=0$, a contradiction. Combining $x^r-y^r = (x+1)^r-(y+1)^r$ and $x^r+y^r = (x+1)^r+(y+1)^r$, we obtain $x^r = (x+1)^r$ and $y^r = (y+1)^r$, which implies $x=x+1$ and $y=y+1$, a contradiction.

If $(x,y)\in S_{01}\times S_{01}$, then the second equation in \eqref{BU_system} leads to $b=0$, a contradiction. Hence there is no solution of \eqref{BU_system} in $S_{01}\times S_{01}$.

If $(x,y)\in S_{01}\times S_{10}$, then equations in \eqref{BU_system} correspond to 
\begin{equation}\label{BU_system S0110}
	\begin{cases}
		2x^r=b,\\
		-2(y+1)^r=b,
	\end{cases}
\end{equation}
and hence $2x^r=-2(y+1)^r=b$. So, we have $\left( \frac{y+1}{x}\right) ^r=-1$. However, we have
\begin{equation*}
	-1=\chi(-1)=\chi\left( \left( \frac{y+1}{x}\right)^r \right) = \left( \chi\left( \frac{y+1}{x}\right) \right) ^r  = \left(\frac{\chi(y+1)}{\chi(x)} \right)^r=1, 
\end{equation*}
a contradiction. Hence, there is no solution of \eqref{BU_system} in $S_{01}\times S_{10}$.

The proof for $B_{1001}(b)=0$ and $B_{1010}(b)=0$ are very similar to $B_{0110}(b)=0$ and $B_{0101}(b)=0$, respectively, and we omit it here.
\end{proof}

By Lemma \ref{BU_lemma S11} and Lemma \eqref{BU_lemma nosol}, we obtain 
\begin{equation}\label{beta1b_eqn}
	\beta_{F_{r,1}}(1,b)=B_{0001}(b)+B_{0010}(b)+B_{0100}(b)+B_{1000}(b).
\end{equation}

\begin{lemma}\label{BU_lemma S01}
	Let $p^n\equiv 7\pmod{8}$ and $b\in \Fpnmul$. Then,
	\begin{align*}
		B_{0001}(b)&=
		\begin{cases}
			1 &\text{ if } \chi(b)=1,\ \chi(b^2-4)=1,\ \chi(b-(b^2-4)^r)=-1,\\
			0 &\text{ otherwise, } 
		\end{cases}\\
		B_{0100}(b)&=
		\begin{cases}
			1 &\text{ if } \chi(b)=-1,\ \chi(b^2-4)=1,\ \chi(b+(b^2-4)^r)=1,\\
			0 &\text{ otherwise. } 
		\end{cases}
	\end{align*}
\end{lemma}
\begin{proof}
If $(x,y)\in S_{00}\times S_{01}$, then equations in \eqref{BU_system} correspond to 
\begin{equation*}
	\begin{cases}
		x^r-y^r = \frac{b}{2},\\
		(x+1)^r = \frac{b}{2}.
	\end{cases}
\end{equation*} 
From the second equation, one has $x = (x+1)-1 = \frac{b^2}{4}-1 = \frac{b^2-4}{4}$. Hence, $\chi(x)=\chi(x+1)=1$ if and only if $\chi(b) = \chi(b^2-4) = 1$. By multiplying $x^r + y^r$ on the both sides of the first equation one has
\begin{equation*}
	x^r + y^r = \frac{2}{b}(x-y),
\end{equation*}
which implies $2x^r = \frac{b}{2} + \frac{2}{b}(x-y)$. Hence, one gets a unique 
\begin{equation*}
	y = \frac{b^2}{4} + x - bx^r = \frac{2b^2-4 -2b(b^2-4)^r}{4} = \left( \frac{b-(b^2-4)^r}{2}\right)^2,
\end{equation*}
and $y+1 = \frac{b}{2}(b-(b^2-4)^r)$. Since $\chi\left( \frac{b}{2} \right)=1$, one gets $\chi(y+1)=-1$ if and only if $\chi(b-(b^2-4)^r)=1$.

The proof for $B_{0100}(b)$ is very similar to $B_{0001}(b)$, and we omit it here.
\end{proof}

\begin{lemma}\label{BU_lemma S01+}
	Let $p^n\equiv 7\pmod{8}$ and $b\in \Fpnmul$. Then
	\begin{equation}\label{BU01}
		B_{0001}(b)+B_{0100}(b)=
		\begin{cases}
			1 &\text{ if }\chi(b^2-4)=1,\ \chi(b^2+2b)=-1,\\
			0 &\text{ otherwise.}
		\end{cases}
	\end{equation}
\end{lemma}
\begin{proof}
	Since $\chi(b+(b^2-4)^r)\chi(b-(b^2-4)^r)= \chi(4)=1$, the result of Lemma \ref{BU_lemma S01} is summarized as follows.
	\begin{equation}\label{BU01_eqn}
		B_{0001}(b)+B_{0100}(b)=
		\begin{cases}
			1 &\text{ if }\chi(b^2-4)=1,\ \chi(b)\chi(b+(b^2-4)^r)=-1,\\
			0 &\text{ otherwise.}
		\end{cases}
	\end{equation}
	Assume that $\chi(b^2-4)=\chi(b+2)\chi(b-2)=1$. Then, $\chi(b+2)=\chi(b-2)$. Observe that 
	\begin{equation*}
		\left( (b+2)^r + (b-2)^r\right)^2 = 2\left( b\cdot \chi(b+2) + (b^2-4)^r \right). 
	\end{equation*}
	Since $\chi(2)=1$, if $\chi(b+2)=1$ then we have $\chi(b+(b^2-4)^r)=1$. If $\chi(b+2)=-1$, then we obtain
	\begin{equation*}
		\chi(b+(b^2-4)^r) = \chi(b-(b^2-4)^r) = -\chi(-b+(b^2-4)^r) = -1.
	\end{equation*} 
	Therefore, we have $\chi(b+2) = \chi(b+(b^2-4)^r)$, which completes the proof. 
\end{proof}

\begin{lemma}\label{BU_lemma S10}
	Let $p^n\equiv 7\pmod{8}$ and $b\in \Fpnmul$. Then 
	\begin{align*}
		B_{0010}(b)&=
		\begin{cases}
			1 &\text{ if } \chi(b)=1,\ \chi(b^2+4)=1,\ \chi(b-(b^2+4)^r)=-1,\\
			0 &\text{ otherwise, } 
		\end{cases}\\
		B_{1000}(b)&=
		\begin{cases}
			1 &\text{ if } \chi(b)=-1,\ \chi(b^2+4)=1,\ \chi(b+(b^2+4)^r)=1,\\
			0 &\text{ otherwise. } 
		\end{cases}
	\end{align*}
	In particular, we have $B_{0010}(b)+B_{1000}(b)\le 1$.
\end{lemma}

\begin{proof}
If $(x,y)\in S_{00}\times S_{10}$, then equations in \eqref{BU_system} lead to 
\begin{equation}\notag
	\begin{cases}
		x^r=\frac{b}{2},\\
		(x+1)^r-(y+1)^r=\frac{b}{2}.
	\end{cases}
\end{equation} 
From the first equation, one has $\chi\left(\frac b 2 \right)=1$, $x = \frac{b^2}{4}$ and $x+1 = \frac{b^2 +4}{4}$. Hence, one gets $\chi(x) = \chi(x+1) = 1$ if and only if $\chi(b) = \chi(b^2+4) = 1$.
By multiplying $(x+1)^r+(y+1)^r$ to the second equation, one has 
$$(x+1)^r+(y+1)^r=\frac{2}{b}(x-y), $$
which implies $2(x+1)^r=\frac{b}{2}+\frac{2}{b}(x-y)$. Therefore one gets  a unique $y=\frac{b^2}{4}+x-b(x+1)^r=\frac{b}{2}(b-(b^2+4)^r)$. Using $\chi(y)=-1$ and $\chi\left( \frac{b}{2} \right)=1$, one gets $\chi(b-(b^2+4)^r)=-1$. 

The proof for $B_{1000}(b)$ is very similar to $B_{0010}(b)$, and we omit it here.
\end{proof}

\begin{lemma}\label{BU_lemma S10+}
	Let $p^n\equiv 7\pmod{8}$ and $b\in \Fpnmul$. Then
	\begin{equation}\label{BU10}
		B_{0010}(b)+B_{1000}(b)=
		\begin{cases}
			1 &\text{ if }\chi(b^2+4)=1,\ \chi(2+(b^2+4)^r)=1,\\
			0 &\text{ otherwise.}
		\end{cases}
	\end{equation}
\end{lemma}

\begin{proof}
If $\chi(b^2+4)=1$, from $1=\chi(4)=\chi\{((b^2+4)^r+b)((b^2+4)^r-b)\}$, one gets 
$$\chi((b^2+4)^r+b)=\chi((b^2+4)^r-b).$$
Letting $\alpha=(b^2+4)^r+2, \beta=(b^2+4)^r-2$, one has $\chi(\alpha)\chi(\beta)=\chi(b^2)=1$, i.e., $\chi(\alpha)=\chi(\beta)$. Therefore, from 
$$
(\alpha^r+\beta^r)^2=\chi(\alpha)\alpha+\chi(\beta)\beta+2\chi(b)b=\chi(\alpha)\cdot 2(b^2+4)^r+2\chi(b)b,
$$
one always has  $1=\chi\left(\chi(\alpha) (b^2+4)^r+\chi(b)b\right)=\chi(\alpha)\chi\left((b^2+4)^r+\frac{\chi(b)}{\chi(\alpha)}b\right)$. That is, by recalling $\alpha=(b^2+4)^r+2$, one has $\chi((b^2+4)^r+2)=\chi((b^2+4)^r\pm b)$ independent of the sign of $\chi(b)$.
\end{proof}

	\begin{lemma}\label{supersingular_lemma}
		Let $p$  be an odd prime such that $\chi(-2)=-1$ and let $g(x)=x^2+ax+\frac18a^2 \in \Fpn[x]$ with $a\neq 0$. Then, 
		$$
		\sum_{x\in \Fpn} \chi (xg(x))=0.
		$$
	\end{lemma}
	
	\begin{proof}
		Letting $l=\sum_{x\in \Fpn} \chi (xg(x))$, 
		\begin{align*}
			l=\sum_{x\in \Fpnmul} \chi (x)\chi(g(x))= \sum_{x\in \Fpnmul} \chi\left(\frac{g(x)}{x}\right) =\chi(c) \sum_{x\in \Fpnmul} \chi\left(\frac{g(x)}{cx}\right),
		\end{align*}
		where $c\in \Fpnmul$ will be determined soon. Writing $t=\frac{g(x)}{cx}$, one has an equivalent expression $0=g(x)-ctx=x^2+(a-ct)x+\frac18a^2$. Denoting the discriminant of the quadratic equation as $D(t)=(a-ct)^2-\frac12a^2=c^2t^2-2act+\frac12a^2$, it is easy to see that the number of solutions $x$ satisfying $\frac{g(x)}{cx}=t$ is given as $1+\chi(D(t))$. Therefore, 
		\begin{align*}
			l&=\chi(c) \sum_{x\in \Fpnmul} \chi\left(\frac{g(x)}{cx}\right)=\chi(c) \sum_{t\in \Fpn} \left(1+\chi(D(t))\right)\chi(t)
			=\chi(c)\sum_{t\in \Fpn}\chi(D(t))\chi(t) \\
			&=\chi(c)\sum_{t\in \Fpn}\chi(t) \chi\left(c^2t^2-2act+\frac12a^2\right)=\chi(c)\sum_{t\in \Fpn}\chi(t) \chi\left(t^2-\frac{2a}{c}t+\frac{a^2}{2c^2}\right).
		\end{align*}
		Finally, putting $c=-2$, one gets 
		$$
		l=\chi(-2)\sum_{t\in \Fpn}\chi(t)\chi\left(t^2+at+\frac18a^2\right)=-\sum_{x\in \Fpn}\chi(xg(x))=-l,
		$$
		which shows $l=0$.
	\end{proof}
	
	\begin{remark}
		The above lemma, as well as Lemma~\ref{chi_lemma x^4-1}, is closely related to the so-called supersingular elliptic curves over $\Fpn$ (see Section 4.6 in\cite{WASH08}), which also have wide applications in post-quantum cryptography~\cite{NIST19} through the use of isogeny problems between supersingular elliptic curves. An elliptic curve $E$ over $\Fpn$ (with $p$ an odd prime) is a nonsingular curve of genus one and can be written in the form $E\colon y^2 = f(x)$, where $f(x) = x^3 + a_2x^2 + a_1x + a_0$ is a cubic polynomial over $\Fpn$ with nonzero discriminant. There are many equivalent definitions for when $E$ is supersingular, one of which states that $E$ is supersingular if and only if $\Sum_{x\in \Fpn} \chi(f(x)) \equiv 0 \pmod{p}$. In this context, Lemma~\ref{supersingular_lemma} implies that the curve $y^2 = x\left(x^2 + ax + \tfrac{1}{8}a^2\right)$ (with $a \in \Fpnmul$) is supersingular when $p^n \equiv 5,7 \pmod{8}$.
		
		Furthermore, the result of Lemma~\ref{chi_lemma x^4-1} arises from the well-known supersingular elliptic curve $E\colon y^2 = 2(x^3 + x)$ (see Proposition 4.37 in \cite{WASH08}), since the curve $C\colon y^2 = x^4 - 1$ is birationally equivalent to $E$ via the map $(x, y) \mapsto \left(\frac{x+1}{x-1}, \frac{2y}{(x-1)^2} \right)$. As the transformation $x \mapsto \frac{x+1}{x-1}$ is a permutation on $\Fpn \setminus \{1\}$, we easily obtain 
		$$\Sum_{x\in \Fpn} \chi(x^4-1)=\Sum_{x\neq 1} \chi(x^4-1)=\Sum_{x\neq 1} \chi(2(x^3+x))=\left(\Sum_{x\in \Fpn} \chi(2(x^3+x))\right) - \chi(4)=-1,$$ 
		where the supersingular property of $E$ is used in the final equality.
	\end{remark}
	
	\begin{lemma}\label{BS_chi_lemma}
		If $\chi(-2)=-1$, then 
		\begin{equation*}
			\Sum_{x\in \Fpn} \chi(x^2+1)\chi(x^4-6x^2+1)=-1.
		\end{equation*}
	\end{lemma}
	
	\begin{proof}
		Letting $u=x^2$, 
		\begin{equation*}
			\Sum_{x\in \Fpn} \chi(x^2+1)\chi(x^4-6x^2+1) = 1+2\Sum_{\chi(u)=1, u\in \Fpnmul} \chi(u+1)\chi(u^2-6u+1). 
		\end{equation*}
		By the way, from this point on, $u$ is regarded as an independent variable ranging over $\Fpnmul$, and we have
		\begin{align*}
			2\Sum_{\chi(u)=1, u\in \Fpnmul}& \chi(u+1)\chi(u^2-6u+1) = 
			 \Sum_{u\in \Fpnmul} (1+\chi(u))\chi(u+1) \chi(u^2-6u+1)  \\
			&=  \Sum_{u\in \Fpnmul} \chi(u+1) \chi(u^2-6u+1) +\Sum_{u\in \Fpnmul} \chi(u^2+u) \chi(u^2-6u+1)\\
			&=   \Sum_{u\in \Fpnmul} \chi(u+1) \chi(u^2-6u+1)+\Sum_{u\in \Fpnmul} \chi\left(\frac{1}{u^2}+ \frac{1}{u} \right) \chi\left(\frac{1}{u^2} - 6\frac 1 u + 1 \right)  \\
			&=  2\Sum_{u\in \Fpnmul} \chi(u+1) \chi(u^2-6u+1) = -2 + 2\Sum_{u\in \Fpn} \chi(u+1) \chi(u^2-6u+1).
		\end{align*}
		By Lemma \ref{supersingular_lemma} with $a=-8$,
		\begin{equation*}
			 \Sum_{u\in \Fpn} \chi(u+1) \chi(u^2-6u+1)= \Sum_{x\in \Fpn} \chi(x) \chi(x^2-8x+8) =\Sum_{x\in \Fpn} \chi(x g(x))=0,
		\end{equation*}
		 where $x=u+1$, which completes the proof.
	\end{proof}

\begin{lemma}\label{BS_lemma}
	Let $p^n \equiv 7 \pmod{8}$. Then, 
	\begin{align*}
		\Sum_{\substack{\chi(x^2+1)=1\\ x\in \Fpn}}\chi(1+(x^2+1)^r) = \Sum_{\substack{\chi(x^2+1)=1\\ x\in \Fpn}}\chi(x^2-1)\chi(1+(x^2+1)^r) = -1,\\
		\Sum_{\substack{\chi(x^2+1)=1\\ x\in \Fpn}} \chi(x^2+x) \chi(1+(x^2+1)^r) = \frac{1}{2}\left( -1+\Sum_{x\in \Fpn} \chi(x^4-1)\chi(x^2-2x-1)\right) .
	\end{align*}
\end{lemma}

\begin{proof}
	Let $S=\{x\in \Fpn : \chi(x^2+1)=1\}$. Then, for each $x\in S$, there exists $y\in \Fpnmul$ such that $y^2 = x^2+1$. Since $1=y^2-x^2 = (y+x)(y-x)$, if $t=y+x$, then $\frac{1}{t} = y-x$, and hence $x = \frac{1}{2}\{(y+x)-(y-x)\} = \frac{1}{2}\left( t - \frac{1}{t}\right)$. Because $x= \frac{1}{2}\left( t - \frac{1}{t}\right)$ if and only if $t^2-2xt+1=0$, a map 
	$$t \mapsto \frac{1}{2}\left( t - \frac{1}{t}\right)$$
	is a 2-to-1 map from $\Fpnmul$ to $S$.
	
	If $x= \frac{1}{2}\left( t - \frac{1}{t}\right)$, then
	\begin{align*}
		1 + (x^2+1)^r &= 1 +\left( 1+\frac{1}{4}\left( t - \frac{1}{t}\right)^2\right)^r 
		= 1 +\left( \frac{1}{2}\left( t + \frac{1}{t }\right) \right)^{2r} =  1 +\frac{1}{2}\left( t + \frac{1}{t}\right)\chi\left( \frac{1}{2}\left( t + \frac{1}{t }\right)\right) \\
		& = \frac{1}{2t}\chi\left( t + \frac{1}{t }\right)\left( t^2 +2t\chi\left( t + \frac{1}{t }\right) +1 \right)  = \frac{1}{2t}\chi\left( t + \frac{1}{t }\right) \left( t+\chi\left( t + \frac{1}{t }\right)\right) ^2.
	\end{align*}
	Hence,
	\begin{equation*}
		\chi\left(1 + (x^2+1)^r \right)  = \chi \left( \frac{1}{2t}\right) \chi\left( t + \frac{1}{t }\right)  = \chi (t^2+1).
	\end{equation*}
	Therefore, by Lemma \ref{chi_lemma quad},
	\begin{equation*}
		\Sum_{x\in S}\chi(1+(x^2+1)^r) =  \frac{1}{2}\Sum_{t\in \Fpnmul}\chi(1+t^2) = \frac{1}{2}\left( \Sum_{t\in \Fpn}\chi(1+t^2) - \chi(1)\right) =-1.
	\end{equation*}
	
	Similarly, we apply $x =  \frac{1}{2 }\left( t- \frac {1}{t}\right)$ for the second identity,
	\begin{align*}
		\Sum_{x\in S}&\chi(x^2-1)\chi(1+(x^2+1)^r) = \frac{1}{2}\Sum_{t\in \Fpnmul} \chi\left( \frac{1}{4}\left( t- \frac {1}{t}\right) ^2-1\right)  \chi(t^2+1) \\
		&= \frac{1}{2}\Sum_{t\in \Fpnmul} \chi\left(\frac{t^4-6t^2+1}{4t^2} \right)\chi(t^2+1)
		= \frac{1}{2}\left( -1+\Sum_{t\in \Fpn} \chi(t^2+1) \chi(t^4-6t^2+1) \right)=-1, 
	\end{align*}
	where the last equality is from Lemma \ref{BS_chi_lemma}.
	
	Applying the same substitution for the third identity,
	\begin{align*}
		\Sum_{x\in S}&\chi(x)\chi(x+1)\chi(1+(x^2+1)^r) = \frac{1}{2}\Sum_{t\in \Fpnmul} \chi\left( t- \frac {1}{t}\right) \chi\left( t- \frac {1}{t}+2\right) \chi(t^2+1) \\
		&= \frac{1}{2}\Sum_{t\in \Fpnmul} \chi(t^2-1) \chi(t^2+2t-1) \chi(t^2+1)
		= \frac{1}{2}\left( -1+\Sum_{t\in \Fpn} \chi(t^4-1) \chi(t^2+2t-1) \right)\\
		&= \frac{1}{2}\left( -1+\Sum_{t\in \Fpn} \chi(t^4-1) \chi(t^2-2t-1) \right), 
	\end{align*}
	where the last equality is from replacing $t$ by $-t$.
\end{proof}

Replacing $x$ by $\frac{x}{2}$ in Lemma \ref{BS_lemma}, we have
\begin{align}
	\Sum_{\substack{\chi(x^2+4)=1\\ x\in \Fpn}}\chi(2+(x^2+4)^r) = \Sum_{\substack{\chi(x^2+4)=1\\ x\in \Fpn}}\chi(2+(x^2+4)^r)\chi(x^2-4) = -1, \label{BS_eqn4}\\
	\Sum_{\substack{\chi(x^2+4)=1\\ x\in \Fpn}}\chi(x^2+2x)\chi(2+(x^2+4)^r) = \frac{1}{2}\left( -1+\Sum_{x\in \Fpn}\chi(x^4-1)\chi(x^2-2x-1)\right).  \label{BS_eqn5}
\end{align}
Note that the right-hand side of \eqref{BS_eqn5} is kept in the same form as the right-hand side of the second equation in Lemma \ref{BS_lemma} for later use.

Applying Lemma \ref{Fruproperty2_lemma}, the boomerang spectrum of $F$ is defined to be the multiset $BS_F=\{\nu_i : 0 \le i \le \beta_F\}$, where
\begin{equation*}
	\nu_i = \#\{ b\in \Fpnmul : \beta_F(1,b)=i\}.
\end{equation*}
The following identity for the boomerang spectrum is well-known :
\begin{equation}\label{BS_identity}
	\sum_{i=0}^{\beta_F}\nu_i = p^n-1.
\end{equation}
Now, we are ready to show the boomerang spectrum of $F_{r,1}$, when $p^n \equiv 7 \pmod{8}$.

\begin{theorem}\label{BS_thm}
	Let $p^n \equiv 7 \pmod{8}$. Then, the boomerang spectrum of $F_{r,1}$ is given by 
	\begin{equation*}
		BS_{F_{r,1}} = \left\{ 
		\nu_0 =  \frac{9(p^n+1)+4\Gamma}{16},\quad
		\nu_1 =  \frac{3p^n-13-4\Gamma}{8},\quad 
		\nu_2 =  \frac{p^n+1+4\Gamma}{16} 
		\right\},
	\end{equation*}
	where $\Gamma = \Sum_{\substack{\chi(x)\ne\chi(x^4-1)\\ x\in \Fpn}}\chi(x(x^2-2x-1))$. Furthermore, $\nu_2 > 0$ if $p^n > 790$.
\end{theorem}

\begin{proof}
We denote 
\begin{equation*}
	g_1(x) = x^2-4,\ \ g_2(x) = x^2+2x, \ \ g_3(x) = x^2+4, \ \ g_4(x) = 2+(x^2+4)^r.
\end{equation*}

By Lemma \ref{chi_lemma quad} and \eqref{chi x^4-16}, we have
\begin{equation*}
	\sum_{x\in\Fpn}  \chi(g_1(x)) = \sum_{x\in\Fpn}  \chi(g_2(x)) = \sum_{x\in\Fpn}  \chi(g_3(x)) =  \sum_{x\in\Fpn}  \chi(g_1(x)g_3(x))=-1,\ 
	\sum_{x\in\Fpn}  \chi(g_1(x)g_2(x))=-2.
\end{equation*}
It is easy to see that
\begin{align*}
	\sum_{x\in\Fpn}  \chi(g_1(x) g_2(x) g_3(x)) &=\sum_{x\in\Fpn} \chi((x^2-4)(x^2+2x)(x^2+4)) = -1+\sum_{x\in\Fpn}  \chi((x^2-2x)(x^2+4)) \\
	&= -1+\sum_{y\in\Fpn}  \chi((y^2+2y)(y^2+4)) = -1+\sum_{x\in\Fpn}  \chi(g_2(x) g_3(x)).
\end{align*}
Using a similar approach,
\begin{align*}
	\sum_{x\in\Fpn}  \chi(g_1(x) g_2(x) g_4(x)) &= -\chi(1+2^r)+\sum_{x\in\Fpn}  \chi(g_2(x) g_4(x)),\\
	\sum_{x\in\Fpn}  \chi(g_1(x) g_2(x) g_3(x) g_4(x)) &=  -\chi(1+2^r) + \sum_{x\in\Fpn}  \chi(g_2(x) g_3(x) g_4(x)).
\end{align*}
By \eqref{BS_eqn4} and \eqref{BS_eqn5}, we have
\begin{align*}
	\sum_{x\in \Fpn}\chi(g_4(x))  + \sum_{x\in \Fpn}\chi(g_3(x)g_4(x))  = 
	\sum_{x\in \Fpn}\chi(g_1(x)g_4(x))  + \sum_{x\in \Fpn}\chi(g_1(x)g_3(x)g_4(x)) =  -2,\\
	\sum_{x\in \Fpn}\chi(g_2(x)g_4(x))  + \sum_{x\in \Fpn}\chi(g_2(x)g_3(x)g_4(x)) = -1 + \Sum_{x\in \Fpn}\chi(x^4-1)\chi(x^2-2x-1).
\end{align*}

Let $z=\frac{x(x+2)}{x^2+4}$. Then $z=1$ if and only if $x=2$. If $z\ne 1$ then $(z-1)x^2-2x+4z=0$. The discriminant can be computed as $(-2)^2-16z(z-1) = 4(1+4z-4z^2)$. Hence, 
\begin{align*}
	\sum_{x\in\Fpn} \chi(g_2(x) g_3(x)) &= \sum_{x\in\Fpn} \chi(x(x+2)(x^2+4)) =1+\sum_{x\in\Fpn\setminus\{2\}} \chi(x(x+2)(x^2+4)) \\
	&= 1+\sum_{z\in \Fpn\setminus\{1\}} \chi(z)(1+\chi(1+4z-4z^2)) = -1+\sum_{z\in\Fpn}\chi(z(1+4z-4z^2))\\ 
	&= -1+\sum_{w\in\Fpn}\chi(w(1+2w-w^2)) =-1-\sum_{w\in\Fpn}\chi(w(w^2-2w-1)).
\end{align*}

We denote 
\begin{align*}
	A_1 &= \{ x\in \Fpnmul : B_{0001}(x)+B_{0100}(x)=1\} = \{ x\in \Fpnmul : \chi(g_1(x))=1,\ \chi(g_2(x))=-1\}, \\
	A_2 &= \{ x\in \Fpnmul : B_{0010}(x)+B_{1000}(x)=1\} =\{ x\in \Fpnmul : \chi(g_3(x)))=\chi(g_4(x))=1\}. 
\end{align*}
Then, using the above identities,
\begin{align*}
	\nu_2 =&\#(A_1 \cap A_2) = \frac{1}{16}\sum_{x\in \Fpnmul\setminus\{\pm 2\}} (1+\chi(g_1(x)))(1-\chi(g_2(x)))(1+\chi(g_3(x)))(1+\chi(g_4(x)))\\
	=& \frac{p^n -3}{16} +  \frac{1}{8}\left(\sum_{x\in \Fpn}\chi(x(x^2-2x-1)) +1-\sum_{x\in \Fpn}\chi(x^4-1)\chi(x^2-2x-1)  \right) \\
	=& \frac{p^n+1}{16}+\frac{1}{4}\sum_{\chi(x)\ne\chi(x^4-1)}\chi(x(x^2-2x-1))  =\frac{p^n+1 + 4\Gamma}{16}. 
\end{align*}
Also, from the previous results, we obtain
\begin{align*}
	\# A_1 &= \frac{1}{4}\sum_{x\in \Fpn \setminus \{0,\pm 2\}} (1+\chi(g_1(x))) (1-\chi(g_2(x))) =  \frac{p^n+1}{4},\\
	\# A_2 &= \frac{1}{4}\sum_{x\in \Fpnmul} (1+\chi(g_3(x))) (1+\chi(g_4(x))) = \frac{p^n-7}{4}.
\end{align*}
Hence,
\begin{align*}
	\nu_1 =&\#((A_1\cup A_2)\setminus(A_1\cap A_2))  = \# A_1 + \#A_2 - 2\nu_2\\
	&= \frac{p^n+1}{4} + \frac{p^n-7}{4} - 2\left( \frac{p^n +1+4\Gamma}{16}  \right)  = \frac{3p^n-13-4\Gamma}{8}.
\end{align*}
By \eqref{BS_identity},
\begin{equation*}
	\nu_0 = p^n-1-\nu_2-\nu_1 = \frac{9(p^n+1)+4\Gamma}{16}.
\end{equation*}

By Lemma \ref{chi_lemma inequal}, we have
\begin{align*}
	|\Gamma| &= \left|\sum_{x\in \Fpn}\chi(x(x^2-2x-1)) -\sum_{x\in \Fpn}\chi(x^4-1)\chi(x^2+2x-1) -1 \right|\\
	&\le \left|\sum_{x\in \Fpn}\chi(x(x^2-2x-1))\right| +\left|\sum_{x\in \Fpn}\chi(x^4-1)\chi(x^2+2x-1)  \right|+1 \\
	&\le 2\sqrt{p^n} +  5\sqrt{p^n}+ 1=7\sqrt{p^n}+1.
\end{align*}
Hence, $-7\sqrt{p^n}-1 \le \Gamma \le 7\sqrt{p^n}+1$. Therefore, we have
\begin{equation*}
	\nu_2 = \frac{p^n+1+4\Gamma}{16}\ge \frac{p^n-28\sqrt{p^n}-3}{16}.
\end{equation*}
We can check that $p^n-28\sqrt{p^n}-3 > 0$ if $p^n > 790$, via SageMath.
\end{proof}

We confirm that the above theorem holds for $p^n < 10000$ via SageMath.  Table \ref{table_bs} presents $BS_{F_{r,1}}$ for $p^n < 790$, which is consistent with the result established in Theorem \ref{BS_thm}. Combining these results, we conclude that $\beta_{F_{r,1}}=2$ for all $p^n\equiv 7\pmod{8}$, except for $p^n = 7$ and $31$.

\begin{table}
	\begin{center}
		\begin{tabular}{ccc|ccc}
			\hline $p^n$ & $\Gamma$ & $BS_{F_{r,1}}$ &  $p^n$\ & $\Gamma$ & $BS_{F_{r,1}}$ \\
			\hline $7$ & $-2$ & $\{ \nu_0 = 4,\ \nu_1 = 2 \}$ & $343$ & $10$ & $\{\nu_0 = 196,\ \nu_1 = 122,\ \nu_2 = 24 \}$\\
			$23$ &  $2$ & $\{ \nu_0 = 14,\ \nu_1 = 6,\ \nu_2 = 2 \}$ & $31$ & $-8$ & $\{\nu_0 = 16,\ \nu_1 = 14 \}$\\
			$47$ & $4$ & $\{ \nu_0 = 28,\ \nu_1 = 14,\ \nu_2 = 4 \}$ & $71$ & $6$ & $\{\nu_0 = 42,\ \nu_1 = 22,\ \nu_2 = 6 \}$\\
			$79$ & $-4$ & $\{ \nu_0 = 44,\ \nu_1 = 30,\ \nu_2 = 4 \}$ & $103$ & $-4$ &  $\{\nu_0 = 58,\ \nu_1 = 38,\ \nu_2 = 6 \}$\\
			$127$ & $16$ &  $\{ \nu_0 = 76,\ \nu_1 = 38,\ \nu_2 = 12 \}$ & $151$ & $2$ &  $\{ \nu_0 = 86,\ \nu_1 = 54,\ \nu_2 = 10 \}$\\
			$167$ & $-10$ &  $\{ \nu_0 = 92,\ \nu_1 = 66,\ \nu_2 = 8 \}$ & $191$ & $-16$ &  $\{ \nu_0 = 104,\ \nu_1 = 78,\ \nu_2 = 8 \}$\\
			$199$ & $-2$ &  $\{ \nu_0 = 112,\ \nu_1 = 74,\ \nu_2 = 12 \}$ & $223$ & $24$ &  $\{ \nu_0 = 132,\ \nu_1 = 70,\ \nu_2 = 20 \}$\\
			$239$ & $-12$ &  $\{ \nu_0 = 132,\ \nu_1 = 94,\ \nu_2 = 12 \}$ & $263$ & $6$ &  $\{ \nu_0 = 150,\ \nu_1 = 94,\ \nu_2 = 18 \}$\\
			$271$ & $-4$ &  $\{ \nu_0 = 152,\ \nu_1 = 102,\ \nu_2 = 16 \}$ & $311$ & $-14$ &  $\{ \nu_0 = 172,\ \nu_1 = 122,\ \nu_2 = 16 \}$\\
			$359$ & $-2$ &  $\{ \nu_0 = 202,\ \nu_1 = 134,\ \nu_2 = 22 \}$ & $367$ & $4$ &  $\{ \nu_0 = 208,\ \nu_1 = 134,\ \nu_2 = 24 \}$\\
			$383$ & $-16$ &  $\{ \nu_0 = 212,\ \nu_1 = 150,\ \nu_2 = 20 \}$ & $431$ & $52$ &  $\{ \nu_0 = 256,\ \nu_1 = 134,\ \nu_2 = 40 \}$\\
			$439$ & $18$ &  $\{ \nu_0 = 252,\ \nu_1 = 154,\ \nu_2 = 32 \}$ & $463$ & $-20$ &  $\{ \nu_0 = 256,\ \nu_1 = 182,\ \nu_2 = 24 \}$\\
			$479$ & $-40$ &  $\{ \nu_0 = 260,\ \nu_1 = 198,\ \nu_2 = 20 \}$ & $487$ & $-10$ &  $\{ \nu_0 = 272,\ \nu_1 = 186,\ \nu_2 = 28 \}$\\
			$503$ & $10$ &  $\{ \nu_0 = 286,\ \nu_1 = 182,\ \nu_2 = 34 \}$ & $599$ & $-6$ &  $\{ \nu_0 = 336,\ \nu_1 = 226,\ \nu_2 = 36 \}$\\
			$607$ & $-40$ &  $\{ \nu_0 = 332,\ \nu_1 = 246,\ \nu_2 = 28 \}$ & $631$ & $-22$ &  $\{ \nu_0 = 350,\ \nu_1 = 246,\ \nu_2 = 34 \}$\\
			$647$ & $6$ &  $\{ \nu_0 = 366,\ \nu_1 = 238,\ \nu_2 = 42 \}$ & $719$ & $12$ &  $\{ \nu_0 = 408,\ \nu_1 = 262,\ \nu_2 = 48 \}$\\
			$727$ & $-6$ &  $\{ \nu_0 = 408,\ \nu_1 = 274,\ \nu_2 = 44 \}$ & $743$ & $22$ &  $\{ \nu_0 = 424,\ \nu_1 = 266,\ \nu_2 = 52 \}$\\
			\hline
		\end{tabular}
		\caption{Boomerang spectrum $BS_{F_{r,1}}$ when $p^n < 790$ with $p^n \equiv 7 \pmod{8}$.}\label{table_bs}
	\end{center}
\end{table}

\section{Differential Uniformity when $u\ne \pm 1$}\label{sec u!=pm1}

In this section, we study the differential uniformity of $F_{r,u}$ when $u \in \Fpnmul \setminus \{1, -1\}$. In this case, we consider to find the number of solutions of 
\begin{equation}\label{DU_eqn}
	b=F_{r,u}(x+1)-F_{r,u}(x) = (x+1)^r(1+u\chi(x+1)) - x^r (1+u\chi(x)).
\end{equation}
Denote $D_{ij}^u(b)$ be the number of solutions of \eqref{DU_eqn} in $S_{ij}$ where $i,j\in\{0,1\}$.

\begin{lemma}\label{DU_lemma b=0}
	Let $u\in \Fpnmul \setminus \{\pm1\}$. Then $\delta_{F_{r,u}}(1,0)\le 2$.
\end{lemma}
\begin{proof}
	If $x=0$, then \eqref{DU_eqn} implies $b=1+u=0$, which contradicts to $u\ne -1$. If $x=-1$, then \eqref{DU_eqn} implies $b=-(-1)^r(1-u)=0$, which contradicts to $u\ne 1$. Hence, $x=0,-1$ are not solutions of \eqref{DU_eqn}.
	
	If $x\in S_{00}\cup S_{11}$, then \eqref{DU_eqn} is equivalent to $(x+1)^r=x^r$. Squaring on the both sides, we have $x+1=x$, a contradiction. Thus, $D_{00}^u(0)=D_{11}^u(0)=0$
	
	If $x\in S_{01}$, then \eqref{DU_eqn} implies
	\begin{equation*}
		(1+u)x^r=(1-u)(x+1)^r.
	\end{equation*}
	Squaring on the both sides of the above equation, we have
	\begin{equation*}
		(1+u)^2x=-(1-u)^2(x+1),
	\end{equation*}
	which implies $x=-\frac{(1-u)^2}{2(1+u^2)}$. Hence, $D_{01}^u(0)\le 1$. Similarly, we have $D_{10}^u(0)\le 1$. 
	
	Therefore, we obtain $\delta_{F_{r,u}}(1,0)=D_{00}^u(0)+D_{11}^u(0)+D_{01}^u(0)+D_{10}^u(0)\le 2$, which completes the proof.
\end{proof}

From now, we consider the case that $b\ne 0$.

\begin{lemma}\label{DU_lemma S0011}
	Let $u\in \Fpnmul \setminus \{\pm1\}$ and $b\in \Fpnmul$. Then, 
	\begin{align*}
		D_{00}^u(b)&=
		\begin{cases}
			1 &\text{ if }\chi((1+u)^2+b^2)=\chi((1+u)^2-b^2)=\chi(2b(1+u)),\\
			0 &\text{ otherwise,}
		\end{cases}
		\\ 
		D_{11}^u(b)&=
		\begin{cases}
			1 &\text{ if }\chi((1-u)^2+b^2)=\chi((1-u)^2-b^2)=-\chi(b(1-u)),\\
			0 &\text{ otherwise.}
		\end{cases}
	\end{align*}
\end{lemma}

\begin{proof}
	If $x\in S_{00}$, then \eqref{DU_eqn} is equivalent to 
	\begin{equation}\label{DU_eqn S00}
		\frac{b}{1+u} = (x+1)^r- x^r.  
	\end{equation}
	Multiplying $(x+1)^r+ x^r$ on the both sides of \eqref{DU_eqn S00}, we have 
	\begin{equation*}
		\frac{1+u}{b}=(x+1)^r+ x^r.
	\end{equation*}
	Hence, we obtain
	\begin{equation*}
		x^r = \frac{1}{2}\left( \frac{1+u}{b} - \frac{b}{1+u} \right) = \frac{(1+u)^2-b^2}{2b(1+u)}, \ \ (x+1)^r = \frac{1}{2}\left( \frac{1+u}{b} + \frac{b}{1+u} \right) = \frac{(1+u)^2 + b^2}{2b(1+u)}.
	\end{equation*}
	Squaring on the both sides of the above equations, we have
	\begin{equation*}
		x = \frac{((1+u)^2-b^2)^2}{4b^2(1+u)^2}, \ \ x+1 = \frac{((1+u)^2-b^2)^2}{4b^2(1+u)^2} + 1 =  \frac{((1+u)^2 + b^2)^2}{4b^2(1+u)^2} , 
	\end{equation*}
	so that $\chi(x)=\chi(x+1)=1$. Substituting $x=x_{00} = \frac{((1+u)^2-b^2)^2}{4b^2(1+u)^2}$ in \eqref{DU_eqn S00} 
	\begin{align*}
		 (x_{00}+1)^r - x_{00}^r &= \left( \frac{ \left( 	(1+u)^2+b^2\right)^2}{4b^2(1+u)^2}\right)^r - \left( \frac{\left( (1+u)^2-b^2\right)^2}{4b^2(1+u)^2}\right)^r  \\
		&= \frac{(1+u)^2+b^2}{2b(1+u)}\chi\left( \frac{(1+u)^2+b^2}{2b(1+u)} \right) - \frac{(1+u)^2-b^2}{2b(1+u)}\chi\left( \frac{(1+u)^2-b^2}{2b(1+u)} \right)   \\
		&=\frac{1+u}{2b}\left( \chi\left( \frac{(1+u)^2+b^2}{2b(1+u)}\right)  -\chi\left( \frac{(1+u)^2-b^2}{2b(1+u)}\right) \right) \\
		&+\frac{b}{2(1+u)}\left( \chi\left( \frac{(1+u)^2+b^2}{2b(1+u)}\right)  +\chi\left( \frac{(1+u)^2-b^2}{2b(1+u)}\right) \right).
	\end{align*}
	We can see that $x_{00}$ is a solution of \eqref{DU_eqn S00} if and only if 
	\begin{equation}\label{DU_eqn S00 sol}
		\chi\left( \frac{(1+u)^2+b^2}{2b(1+u)}\right)=\chi\left( \frac{(1+u)^2-b^2} {2b(1+u)}\right)=1,
	\end{equation}
	or equivalently,
	\begin{equation*}
		\chi((1+u)^2+b^2)=\chi((1+u)^2-b^2)=\chi(2b(1+u)).
	\end{equation*}

	The proof for $D_{11}^u(b)$ is similar to $D_{00}^u(b)$ and we omit it here.
\end{proof}

\begin{lemma}\label{DU_lemma S01}
	Let $u\in \Fpnmul \setminus \{\pm1\}$ and $b\ne 0$. Then, $D_{01}^u(b)\le 2$. Furthermore, $D_{01}^u(b) = 2$ if and only if $\chi(b^2+2(1+u^2))=-1$ and 
	\begin{align*}
		\chi(b(1-u)+(1+u)R_1) = \chi(b(1-u)-(1+u)R_1) &= \chi(1+u^2),\\
		\chi(b(1+u)+(1-u)R_1) = \chi(b(1+u)-(1-u)R_1) &= -\chi(2(1+u^2)),
	\end{align*}
	where $R_1  = (b^2+2(1+u^2))^r$.
\end{lemma}

\begin{proof}
If $x\in S_{01}$, then \eqref{DU_eqn} is equivalent to 
\begin{equation}\label{DU_eqn S01}
	b = (1- u)(x+1)^r- (1+u)x^r.  
\end{equation}
Multiplying $(1- u)(x+1)^r+ (1+u)x^r$ on the both sides of \eqref{DU_eqn S01}, we have 
\begin{equation*}
	b(1-u)(x+1)^r + b(1+u)x^r = -2(1+u^2)x - (1-u)^2.
\end{equation*}
Subtracting $b$ times \eqref{DU_eqn S01} from the last equation, we obtain 
\begin{equation*}
	2b(1+u)x^r = - 2(1+u^2)x - ((1-u)^2+b^2).
\end{equation*}
Squaring on the both sides of the above equation, we have the following quadratic equation
\begin{equation}\label{DU_eqn S01 quad}
	0=4(1+u^2)^2x^2+4x\left( -2ub^2 +(1-u)^2(1+u^2) \right) +\left( b^2 +(1-u)^2\right)^2. 
\end{equation}
Thus, we have $D_{01}^u(b) \le 2$. The discriminant of \eqref{DU_eqn S01 quad} is 
\begin{align*}
	&16\left( - 2ub^2 +(1-u)^2(1+u^2) \right)^2-16(1+u^2)^2\left( b^2 +(1-u)^2\right)^2\\
	&= -16b^2(1-u^2)^2\left( b^2+2(1+u^2) \right).
\end{align*}
So, \eqref{DU_eqn S01 quad} has two solutions if and only if 
\begin{equation*}
	\chi\left( b^2+2(1+u^2) \right)=-1.
\end{equation*}
If then, two solutions are
\begin{equation}
	x_{01}=\frac{2\left(2ub^2-(1-u)^2(1+u^2)\right)+2b\epsilon(1-u^2)R_1}{4(1+u^2)^2}=\frac{\left( b(1+u)+\epsilon(1-u)R_1\right) ^2}{4(1+u^2)^2}, \label{DU_eqn sol S01}
\end{equation}
where $\epsilon\in \{1,-1\}$. Then, we have
\begin{equation*}
	x_{01}+1=\frac{2\left(2ub^2+(1+u)^2(1+u^2)\right)+2b\epsilon(1-u^2)R_1}{4(1+u^2)^2} \\
	=-\frac{\left( b(1-u)-\epsilon(1+u)R_1\right) ^2}{4(1+u^2)^2}.
\end{equation*}
We can see that $\chi(x_{01})=1$ and $\chi(x_{01}+1)=-1$. Substituting $x=x_{01}$ in the right hand side of \eqref{DU_eqn S01},
\begin{align*}
	(1- u)&(x_{01}+1)^r- (1+u)x_{01}^r = (1-u)\left( -\frac{\left( b(1-u)-\epsilon(1+u)R_1\right) ^2}{4(1+u^2)^2}\right) ^r - (1+u) \left( \frac{\left( b(1+u)+\epsilon(1-u)R_1\right) ^2}{4(1+u^2)^2} \right) ^r\\
	&=\frac{b}{2(1+u^2)}\left( (-1)^r (1-u)^2\chi\left( \frac{ b(1-u)-\epsilon(1+u)R_1}{2(1+u^2)}\right) -(1+u)^2\chi\left( \frac{ b(1+u)+\epsilon(1-u)R_1}{2(1+u^2)}\right) \right) \\
	&-\epsilon\frac{(1-u^2)R_1}{2(1+u^2)}\left( (-1)^r \chi\left( \frac{ b(1-u)-\epsilon(1+u)R_1}{2(1+u^2)}\right) +\chi\left( \frac{ b(1+u)+\epsilon(1-u)R_1}{2(1+u^2)}\right) \right).
\end{align*}

One can easily verify that the above equation holds if and only if 
	\begin{equation*}
		\chi\left( \frac{ b(1-u)-\epsilon(1+u)R_1}{2(1+u^2)}\right) = (-1)^r,\ \chi\left( \frac{ b(1+u)+\epsilon(1-u)R_1}{2(1+u^2)}\right)=-1,
	\end{equation*}
while all other cases lead to a contradiction.
\end{proof}

\begin{lemma}\label{DU_lemma S10}
	Let $u\in \Fpnmul \setminus \{\pm1\}$ and $b\ne 0$. Then, $D_{10}^u(b)\le 2$. Furthermore, $D_{10}^u(b) = 2$ if and only if $\chi(b^2-2(1+u^2))=-1$ and 
	\begin{align*}
		\chi(b(1-u)+(1+u)R_2) = \chi(b(1-u)-(1+u)R_2) &= -\chi(1+u^2),\\
		\chi(b(1+u)+(1-u)R_2) = \chi(b(1+u)-(1-u)R_2) &= \chi(2(1+u^2)),
	\end{align*}
	where $R_2  = (b^2-2(1+u^2))^r $. In particular, if $b^2\in \{(1+u)^2, (1-u)^2\}$, then $D_{10}^u(b)\le 1$. If $D_{10}^u(1+u)= 1$, then $\chi(2(1+u^2))=\chi(u)$. If $D_{10}^u((-1)^{r+1}(1-u))= 1$, then $\chi(2(1+u^2))=-\chi(u)$. 
\end{lemma}

\begin{proof}
If $x\in S_{10}$, then \eqref{DU_eqn} is equivalent to 
\begin{equation}\label{DU_eqn S10}
	b=(1+ u)(x+1)^r- (1-u)x^r .  
\end{equation}
Using the same idea as in Lemma \ref{DU_lemma S01}, we have the following quadratic equation
\begin{equation}\label{DU_eqn S10 quad}
	0=4(1+u^2)^2x^2+4x\left( -2ub^2 +(1+u)^2(1+u^2) \right) +\left(  (1+u)^2 -b^2\right)^2. 
\end{equation}
Thus, we have $D_{10}^u(b) \le 2$. The discriminant of \eqref{DU_eqn S10 quad} is 
\begin{equation*}
	-16b^2(1-u^2)^2\left( b^2-2(1+u^2) \right). 
\end{equation*}
So, \eqref{DU_eqn S10 quad} has two solutions if and only if 
\begin{equation*}
	\chi\left( b^2-2(1+u^2) \right)=-1.
\end{equation*}
In this case, as in Lemma \ref{DU_lemma S01}, two solutions are
\begin{equation*}
	x_{10}= -\frac{\left( b(1-u)-\epsilon (1+u)R_2\right)^2}{4(1+u^2)^2},
\end{equation*}
where $\epsilon\in \{1,-1\}$, and we have
\begin{equation*}
	x_{10}+1 =\frac{\left( b(1+u)+\epsilon(1-u)R_2\right) ^2}{4(1+u^2)^2}.
\end{equation*}
We can see that $\chi(x_{10})=-1$ and $\chi(x_{10}+1)=1$. Substituting $x=x_{10}$ in the right hand side of \eqref{DU_eqn S10},
\begin{align*}
	(1+ u)&(x_{10}+1)^r- (1-u)x_{10}^r \\
	&=\frac{b}{2(1+u^2)}\left(  (1+u)^2\chi\left( \frac{ b(1+u)+\epsilon(1-u)R_2}{2(1+u^2)}\right) -(-1)^r(1-u)^2\chi\left( \frac{ b(1-u)-\epsilon(1+u)R_2}{2(1+u^2)}\right) \right) \\
	&+\frac{(1-u^2)R_2 \epsilon}{2(1+u^2)}\left(  \chi\left( \frac{ b(1+u)+\epsilon(1-u)R_2}{2(1+u^2)}\right) +(-1)^r\chi\left( \frac{ b(1-u)-\epsilon(1+u)R_2}{2(1+u^2)}\right) \right).
\end{align*}
Let us denote the resulting expression by $E_{10}$. Then,
\begin{align*}
	\chi\left( \frac{ b(1+u)+\epsilon(1-u)R_2}{2(1+u^2)}\right)=1, \ \chi\left( \frac{ b(1-u)-\epsilon(1+u)R_2}{2(1+u^2)}\right) = (-1)^{r+1}\ &\Rightarrow\ E_{10}=b\\
	\chi\left( \frac{ b(1+u)+\epsilon(1-u)R_2}{2(1+u^2)}\right)=-1, \ \chi\left( \frac{ b(1-u)-\epsilon(1+u)R_2}{2(1+u^2)}\right) = (-1)^r\  &\Rightarrow\ E_{10}=-b\\
	\chi\left( \frac{ b(1+u)+\epsilon(1-u)R_2}{2(1+u^2)}\right)=1,\ \chi\left( \frac{ b(1-u)-\epsilon(1+u)R_2}{2(1+u^2)}\right) = (-1)^r\ &\Rightarrow\ E_{10}=\frac{4ub+2\epsilon (1-u^2)R_2}{2(1+u^2)}\\
	\chi\left( \frac{ b(1+u)+\epsilon (1-u)R_2}{2(1+u^2)}\right)=-1,\ \chi\left( \frac{ b(1-u)-\epsilon (1+u)R_2}{2(1+u^2)}\right) = (-1)^{r+1},\ &\Rightarrow\ E_{10}=\frac{-4ub- 2\epsilon (1-u^2)R_2}{2(1+u^2)}
\end{align*}
In the third and fourth cases above, we obtain
\[
E_{10}=\pm \frac{2ub+\epsilon(1-u^2)R_2}{1+u^2}.
\]
A direct computation shows that $E_{10}=b$ holds if and only if
$b^2=(1+u)^2$ in the third case, and $b^2=(1-u)^2$ in the fourth case.

If $b^2 = (1+u)^2$, then $R_2=(-(1-u)^2)^r=(-1)^r (1-u)\chi(1-u)$. Suppose that $b=1+u$. If $\epsilon = (-1)^{r+1} \chi(1-u)$, then
\begin{equation*}
	x_{10} = -\frac{\left( (1+u)(1-u)-\epsilon (-1)^r(1+u)(1-u)\chi(1-u)\right)^2}{4(1+u^2)^2} = -\frac{(1-u^2)^2}{(1+u^2)^2}.
\end{equation*}
If $\epsilon = (-1)^{r} \chi(1-u)$, then we have $x_{10}=0\not \in S_{01}$. Thus, we can see that $D_{01}^u(1+u)\le 1$. Moreover, if $D_{01}^u(1+u)= 1$, then
\begin{equation*}
	1=\chi\left(\frac{ b(1+u)+\epsilon(1-u)R_2}{2(1+u^2)}\right)=\chi\left( \frac{ (1+u)^2-(1-u)^2}{2(1+u^2)}\right) = \chi\left( \frac{4u}{2(1+u^2)}\right), 
\end{equation*}
or equivalently, $\chi(2(1+u^2))=\chi(u)$. Similarly, we have $D_{01}^u(-(1+u))\le 1$.

If $b^2=(1-u)^2$, then $R_2=(-(1+u)^2)^r=(-1)^r (1+u)\chi(1+u)$. Suppose that $b=(-1)^{r+1}(1-u)$. If $\epsilon = -\chi(1+u)$, then
\begin{equation*}
	x_{10}= -\frac{\left( (-1)^{r+1}(1-u)^2-\epsilon (-1)^r(1+u)^2 \chi(1+u) \right)^2}{4(1+u^2)^2} = -\frac{4u^2}{(1+u^2)^2}.
\end{equation*}
If $\epsilon = \chi(1+u)$, then we have $x_{10}=-1\not \in S_{01}$. Thus, we can see that $D_{01}^u((-1)^{r+1}(1-u)) \le 1$. Furthermore, if $D_{01}^u((-1)^{r+1}(1-u)) = 1$,  then 
\begin{equation*}
	(-1)^{r+1}=\chi\left( \frac{ b(1-u)-\epsilon(1+u)R_2}{2(1+u^2)}\right) =\chi\left( \frac{ (-1)^{r+1}(1-u)^2-(-1)^{r+1}(1+u)^2}{2(1+u^2)}\right)  =(-1)^{r}\chi\left( \frac{4u}{2(1+u^2)}\right), 
\end{equation*}
or equivalently, $\chi(2(1+u^2))=-\chi(u)$. Similarly, we have $D_{01}^u((-1)^{r}(1-u)) \le 1$.
\end{proof}

\begin{lemma}\label{DU_lemma x=0, x=-1}
	Let $u\in \Fpnmul\setminus\{\pm1\}$. Then, $\delta_{F_{r,u}}(1,1+u)\le 4$ and $\delta_{F_{r,u}}(1,(-1)^{r+1}(1-u))\le 4$.
\end{lemma}

\begin{proof}
	Substituting $x=0$ in \eqref{DU_eqn}, we have
	\begin{equation*}
		b=F_{r,u}(1)-F_{r,u}(0) = 1+u.
	\end{equation*}
	Hence, $x=0$ is a solution of $1+u=F_{r,u}(x+1)-F_{r,u}(x)$. Denote $b_0=1+u$.	Since $\chi((1+u)^2+b_0^2)\ne 0 =\chi((1+u)^2-b_0^2)$, we have $D_{00}^u(b_0)=0$, by Lemma \ref{DU_lemma S0011}. 
	We showed in Lemma \ref{DU_lemma S10} that $D_{10}^u(b_0) \le 1$. Moreover, if $D_{10}^u(b_0)=1$ then
	\begin{equation}\label{DU_eqn b=1+u S10}
		\chi(2(1+u^2)=\chi(u).
	\end{equation}
	By Lemma \ref{DU_lemma S0011}, $D_{11}^u(b_0)=1$ if and only if 
	\begin{equation}\label{DU_eqn b=1+u S11}
		\chi(2(1+u^2))=-\chi(u)=-\chi(1-u^2).
	\end{equation}
	Since \eqref{DU_eqn b=1+u S10} and \eqref{DU_eqn b=1+u S11} cannot hold simultaneously, we have $D_{11}^u(b_0)+D_{10}^u(b_0)\le 1$. By Lemma \ref{DU_lemma S01}, $D_{01}^u(b_0) \le 2$. Therefore, we obtain
	\begin{equation*}
		\delta_{F_{r,u}}(1,b_0)=1+D_{00}^u(b_0)+D_{11}^u(b_0)+D_{10}^u(b_0)+D_{01}^u(b_0)\le 4.
	\end{equation*}
	 	
	Substituting $x=-1$ in \eqref{DU_eqn}, we have
	\begin{equation*}
		b=F_{r,u}(0)-F_{r,u}(-1) = (-1)^{r+1}(1-u).
	\end{equation*}
	Hence, $x=-1$ is a solution of $(-1)^{r+1}(1-u) = F_{r,u}(x+1)-F_{r,u}(x)$. Denote $b_{-1} =(-1)^{r+1}(1-u)$. Since $\chi((1-u)^2+b_{-1}^2)\ne 0 =\chi((1-u)^2-b_{-1}^2)$, we have $D_{11}^u(b_{-1})=0$, by Lemma \ref{DU_lemma S0011}. 
	We showed in Lemma \ref{DU_lemma S10} that $D_{10}^u(b_{-1}) \le 1$. Moreover,  if $D_{10}^u(b_{-1})=1$ then 
	\begin{equation}\label{DU_eqn b=(-1)^r (1-u) S10}
		\chi(2(1+u^2))=-\chi(u).
	\end{equation}
	By Lemma \ref{DU_lemma S0011}, $D_{00}^u(b_{-1})=1$ if and only if 
	\begin{equation}\label{DU_eqn b=(-1)^r (1-u) S00}
		\chi(2(1+ u^2))=\chi(u)=-\chi(1-u^2).
	\end{equation}
	Since \eqref{DU_eqn b=(-1)^r (1-u) S10} and \eqref{DU_eqn b=(-1)^r (1-u) S00} cannot hold simultaneously, we have $D_{00}^u(b_{-1})+D_{10}^u(b_{-1})\le 1$. By Lemma \ref{DU_lemma S01}, $D_{01}^u(b_{-1})\le 1$. 
	Therefore, we obtain 
	\begin{equation*}
		\delta_{F_{r,u}}(1,b_{-1})=1+D_{00}^u(b_{-1}) + D_{11}^u(b_{-1})
		+D_{10}^u(b_{-1})+D_{01}^u(b_{-1})\le 4.
	\end{equation*}
	We complete the proof.
\end{proof}

If $b\not \in \{1+u, (-1)^r(1-u)\}$, then 
\begin{equation}\label{DU_formula}
	\delta_{F_{r,u}}(1,b)=D_{00}^u(b)+D_{11}^u(b)+D_{10}^u(b)+D_{01}^u(b).
\end{equation}
By Lemma \ref{DU_lemma S0011}, $D_{00}^u(b)\le 1$ and $D_{11}^u(b)\le 1$. By Lemma \ref{DU_lemma S01} and Lemma \ref{DU_lemma S10}, $D_{01}^u(b)\le 2$ and $D_{10}^u(b)\le 2$, respectively. Hence, applying Lemmas \ref{DU_lemma b=0} and \ref{DU_lemma x=0, x=-1}, we have $\delta_{F_{r,u}}\le 6$. The following lemma shows that $\delta_{F_{r,u}}\le 5$.

\begin{lemma}\label{DU5_lemma}
	Let $b\in \Fpnmul\setminus  \{1+u, (-1)^r(1-u)\}$. If $D_{00}^u(b)=1$ and $D_{01}^u(b)=2$, then $D_{10}^u(b)\le 1$. 
\end{lemma}

\begin{proof}
	Since $D_{01}^u(b)=2$, by Lemma \ref{DU_lemma S01} and using $R_1^2=-b^2+2(1+u^2)$,
	\begin{align*}
		1&=\chi(b(1-u)+(1+u)R_1)\chi(b(1-u)-(1+u)R_1)  =\chi(2(1+u^2)(b^2+(1+u)^2)),\\
		1&=\chi(b(1+u)+(1-u)R_1)\chi(b(1+u)-(1-u)R_1) =\chi(2(1+u^2)(b^2+(1-u)^2)).
	\end{align*}
	Hence, 
	\begin{equation}\label{DU_eqn D01(b)=2}
		\chi(b^2+(1-u)^2)=\chi(b^2+(1+u)^2)=\chi(2(1+u^2)).
	\end{equation}
	Since $D_{00}^u(b)=1$, by Lemma \ref{DU_lemma S0011}, $\chi(b^2+(1+u)^2)=\chi((1+u)^2-b^2)$. Applying \eqref{DU_eqn D01(b)=2}, we have 
	\begin{equation}\label{DU_eqn DU5_lemma}
		\chi((1+u)^2-b^2)=\chi(2(1+u^2)).
	\end{equation}
	Suppose that $D_{10}^u(b)=2$. Then, by Lemma \ref{DU_lemma S10} and using $R_2^2=-\left(b^2-2(1+u^2)\right)$, 
	\begin{align*}
		1&=\chi(b(1-u)+(1+u)R_2)\chi(b(1-u)-(1+u)R_2) =\chi(2(1+u^2)(b^2-(1+u)^2)),\\
		1&=\chi(b(1+u)+(1-u)R_2)\chi(b(1+u)-(1-u)R_2) =\chi(2(1+u^2)(b^2-(1-u)^2)).
	\end{align*}
	Hence, 
	\begin{equation}\label{DU_eqn D10(b)=2}
		\chi(b^2-(1-u)^2)=\chi(b^2-(1+u)^2)=\chi(2(1+u^2)).
	\end{equation}
	which contradicts to \eqref{DU_eqn DU5_lemma}. Therefore, $D_{10}^u(b) \le 1$.
\end{proof}

\begin{lemma}\label{DU4_lemma}
	Let $b\in \Fpnmul\setminus  \{1+u, (-1)^r(1-u)\}$. Assume that $\chi(1+u)=(-1)^r\chi(1-u)$. If $D_{00}^u(b)=D_{11}^u(b)=1$, then $D_{01}^u(b)\le 1$ and $D_{10}^u(b)\le 1$.
\end{lemma}
\begin{proof}
		Suppose that $D_{01}^u(b)=2$. Then, similar to Lemma \ref{DU5_lemma}, \eqref{DU_eqn D01(b)=2} holds, in particular $\chi(b^2+(1-u)^2)=\chi(b^2+(1+u)^2)$. Since $\chi(1+u)=(-1)^r\chi(1-u)=\chi(2(1-u))$,
	\begin{equation*} 
		\chi(b^2+(1+u)^2)=\chi(2b(1+u))\text{ and } \chi(b^2+(1-u)^2)=-\chi(b(1-u))
	\end{equation*} 
	cannot hold simultaneously. Hence, $D_{00}^u(b)\ne 1$ or $D_{11}^u(b) \ne 1$, a contradiction. Therefore,  $D_{01}^u(b)\le 1$.
	
	Suppose that $D_{10}^u(b)=2$. Then, similar to Lemma \ref{DU5_lemma}, \eqref{DU_eqn D10(b)=2} holds, in particular $\chi((1+u)^2-b^2)=\chi((1-u)^2-b^2)$. Since $\chi(1+u)=\chi(2(1-u))$, \begin{equation*} 
		\chi((1+u)^2-b^2)=\chi(2b(1+u))\text{ and } \chi((1-u)^2-b^2)=-\chi(b(1-u))
	\end{equation*} 
	cannot hold simultaneously. Hence, $D_{00}^u(b)\ne 1$ or $D_{11}^u(b) \ne 1$, a contradiction. Therefore,  $D_{10}^u(b)\le 1$.
\end{proof}

Now we are ready to show the main theorem of this section. 

\begin{theorem}\label{DU_thm}
	Let $u\in \Fpnmul\setminus\{\pm1\}$.
	\begin{enumerate}
		\item If $\chi(1+u)=(-1)^r\chi(1-u)$, then $F_{r,u}$ is differentially $4$-uniform.
		\item If $\chi(1+u)\ne(-1)^r\chi(1-u)$, then $F_{r,u}$ is a differentially $5$-uniform permutation. 
	\end{enumerate}
\end{theorem}

\begin{proof}
	By Lemmas \ref{DU_lemma b=0} and \ref{DU_lemma x=0, x=-1}, if $b\in \{0, 1+u,(-1)^r(1-u)\}$, then $\delta_{F_{r,u}}(1,b) \le 4$.
	
	Now we consider that $b\in \Fpnmul \setminus \{1+u, (-1)^{r+1}(1-u)\}$. Suppose that $\chi(1+u)=(-1)^r\chi(1-u)$. If $D_{00}^u(b)+D_{11}^u(b)=2$, or equivalently, $D_{00}^u(b)=D_{11}^u(b)=1$, then we have $D_{01}^u(b)+D_{10}^u(b) \le 2$, by Lemma \ref{DU4_lemma}. Hence, applying \eqref{DU_formula}, we have $\delta_{F_{r,u}}(1,b)\le 4$.
	If $D_{00}^u(b)+D_{11}^u(b)=1$, we assume that $D_{00}^u(b)=1$ and $D_{11}^u(b)=0$. If $D_{01}^u(b)=2$, then we have $D_{10}^u(b) \le 1$, by Lemma \ref{DU5_lemma}. Thus, we have $D_{01}^u(b) + D_{10}^u(b) \le 3$, and hence we have $\delta_{F_{r,u}}(1,b)\le 4$, applying \eqref{DU_formula}. The proof for all other cases are very similar by symmetry, and we omit it here. Therefore, we have $\delta_{F_{r,u}}(1,b)\le 4$ for all $b\in \Fpn$, and hence $\delta_{F_{r,u}}\le 4$.
	
	Assume that $\chi(1+u)\ne(-1)^r\chi(1-u)$. Then, by Theorem \ref{Fru_PP_thm}, $F_{r,u}$ is a PP, since $\gcd\left(\frac{p^n+1}{4}, \frac{p^n-1}{2} \right) =1$. By Lemma \ref{DU5_lemma}, $\delta_{F_{r,u}}\le 5$, and hence $F_{r,u}$ is a differentially $5$-uniform permutation.
\end{proof}

\begin{theorem}
	Let $p^n \equiv 3 \pmod{8}$ with $p > 3$. If $u \in \{\pm \frac{1-2^{r+1}}{3}\}$, then $F_{r,u}$ is a differentially $4$-uniform permutation.
\end{theorem}

\begin{proof}
	By Lemma \ref{Fruproperty1_lemma}, it suffices to consider only the case $u=\frac{1-2^{r+1}}{3}$. Note that $3u^2-2u+3=0$, which is equivalent to $2(1+u^2)=-(1-u)^2$.
	
	Observe that
	\begin{equation*}
		(1+u)(1-u) = \frac{4-2^{r+1}}{3}\cdot \frac{2+2^{r+1}}{3} =\frac{4}{9} (4+2^r)= -2^r \cdot \frac{4}{9}(1+2^r)^2.
	\end{equation*}
	Since $p^n \equiv 3\pmod{8}$, $r$ is odd and $\chi(2)=-1$. Thus, $\chi(-2^r) = \chi(-1)\left(\chi(2)\right)^r = 1$, and hence we have $\chi(1+u) = \chi(1-u)$. By Theorem \ref{Fru_PP_thm}, $F_{r,u}$ is a PP.	
	
	Next, we show that $\delta_{F_{r,u}}(1,b) \le 4$ for all $b\in \Fpn$. By Lemmas \ref{DU_lemma b=0} and \ref{DU_lemma x=0, x=-1}, it suffices to consider $b\in \Fpnmul \setminus \{1+u, 1-u\}$. In this case, $\delta_{F_{r,u}}(1,b)$ is determined by \eqref{DU_formula}. If $D_{01}^u (b) \le 1$ and $D_{10}^u (b) \le 1$, then $\delta_{F_{r,u}}(1,b)\le 4$. So, we consider two cases $D_{01}^u (b) = 2$ or $D_{10}^u (b) =2$.
	
	Assume that $D_{01}^u (b) = 2$. Then, \eqref{DU_eqn D01(b)=2} holds, and hence
	\begin{equation*}
		\chi(b^2+(1-u)^2)=\chi(b^2+(1+u)^2)=\chi(2(1+u^2))=-1.
	\end{equation*}
	By Lemma \ref{DU_lemma S01}, we have $\chi(b^2+2(1+u^2))=\chi(b^2-(1-u)^2)=-1$ and hence $\chi((1-u)^2-b^2)=1$. By Lemma \ref{DU_lemma S0011}, $D_{11}^u (b) =0$. Hence, we have $D_{01}^u (b) + D_{11}^u (b) \le 2$. If $D_{00}^u (b) = 1$, then $D_{10}^u (b) \le 1$, by Lemma \ref{DU5_lemma}. Thus, we get $D_{00}^u (b) + D_{10}^u (b) \le 2$, and therefore $\delta_{F_{r,u}}(1,b)\le 4$, by \eqref{DU_formula}.
	
	The proof for all other cases is very similar by symmetry, and we omit it here. Therefore, $F_{r,u}$ is a differentially $4$-uniform permutation.
\end{proof}

We conducted experiments on the differential uniformity of $F_{r,u}$ using SageMath for $p^n < 10000$, where $u \in \Fpnmul \setminus \{\pm 1\}$. We observed that $\delta_{F_{r,u}} = 4$ for all $u$ satisfying one of the following conditions:
	\begin{equation}\label{DU4_cond}
		\begin{cases}
			\chi(1+u) = (-1)^r \chi(1 - u),\\
			u=\pm \frac{1-2^{r+1}}{3}, \quad p^n \equiv 3 \pmod{8},
		\end{cases}
	\end{equation}
for all $p^n$ with $523 < p^n < 10000$. This suggests that $\delta_{F_{r,u}} = 4$ for all $p^n > 523$, when \eqref{DU4_cond} holds. We also observed that $\delta_{F_{r,u}} = 5$ for all $u \in \Fpnmul \setminus \left\{\pm 1, \pm \frac{1 - 2^{r+1}}{3}\right\}$ satisfying $\chi(1+u) = (-1)^{r+1} \chi(1 - u)$, for all $p^n$ with $4007 < p^n < 10000$. This suggests that $\delta_{F_{r,u}} = 5$ for all $p^n > 4007$, when \eqref{DU4_cond} does not hold. However, the conditions on quadratic characters described in the lemmas presented in this section are significantly more complicated than those in \cite{XBC+24, XLB+24, MW25}, as they involve multiple occurrences of $r$-th powers. Accordingly, it appears very difficult to adapt the techniques from \cite{XBC+24, XLB+24, MW25} to rigorously prove these equalities concerning the differential uniformity of $F_{r,u}$.

\section{Conclusion}\label{sec_con}

In this paper, we investigated the differential and boomerang properties of the binomial function $F_{r,u}(x) = x^r(1 + u \chi(x) )$ over $\Fpn$, where $r = \frac{p^n+1}{4}$ with $p^n \equiv 3 \pmod{4}$ and $u\in \Fpnmul$. Specifically, we proved that $F_{r,u}$ is a differentially $5$-uniform permutation when $\chi(1+u) = (-1)^{r+1} \chi(1-u)$, and a differentially $4$-uniform function when $\chi(1+u) = (-1)^r \chi(1-u)$. Also, we show that $F_{r,u}$ is a differentially $4$-uniform permutation, when $u\in \{\pm \frac{1-2^{r+1}}{3}\}$ and $p^n\equiv 3 \pmod{8}$. Furthermore, we showed that $F_{r,\pm1}$ is locally-PN with boomerang uniformity $0$ when $p^n \equiv 3 \pmod{8}$, and locally-APN with boomerang uniformity at most $2$ when $p^n \equiv 7 \pmod{8}$. To the best of our knowledge, this provides the second known non-PN class with boomerang uniformity $0$, and the first such example in odd characteristic fields with $p > 3$. We also investigated the differential and boomerang spectra of $F_{r,\pm 1}$, and showed that $\beta_{F_{r,1}}=2$ if $p^n \equiv 7 \pmod{8}$ with $p^n \ne 7, 31$. The case studied in this paper constitutes the third known class of functions $F_{r,u}$ whose differential and boomerang properties have been analyzed in detail, following the cases $r = p^n - 2$ and $r = 2$.

Our motivation for focusing on the case $r = \frac{p^n+1}{4}$ stems from the observation that both $x^r$ and $x^r \chi(x)$ are APN in certain cases, and exhibit low differential uniformity even when they are not APN. It is also worth noting that there are several other exponent classes $r$ for which both $x^r$ and $x^r \chi(x)$ are known to exhibit low differential uniformity (see Table 1 of \cite{MW25}). In fact, we have conducted preliminary experiments using SageMath on some of these classes, and observed that many of them indeed possess low differential uniformity, although there are also cases where the differential uniformity exceeds 5. Nevertheless, we expect that some of these classes may contain functions with interesting cryptographic properties, and we plan to explore them further in future work.

\bigskip
\noindent \textbf{Acknowledgements} :
The authors would like to thank the anonymous reviewers for their kind comments and valuable suggestions. This work was supported by the National Research Foundation of Korea (NRF) grant funded by the Korea government (MSIT) (No. 2021R1C1C2003888). Soonhak Kwon was supported by Basic Science Research Program through the National Research Foundation of Korea(NRF) funded by the Ministry of Education (No. RS-2019-NR040081).

\end{document}